\documentclass[11pt,a4paper]{article}
\usepackage[margin=3.5cm]{geometry}

\usepackage{amsmath,amssymb,amsthm}

\usepackage[justification=centering]{caption}
\usepackage{subfig}

\usepackage[ruled,vlined]{algorithm2e} % algorithms
\usepackage{float} % [H] for algorithm placement (preprint)
\usepackage{tabularx, array}
\usepackage{multirow, booktabs} % multicolumns

\usepackage{xcolor}
\usepackage{authblk} % author, affil
\usepackage[backend=biber]{biblatex}
\usepackage{graphicx} % include graphics
\usepackage{setspace}
\providecommand{\keywords}[1]{\textbf{\textit{Keywords~:}} #1}
\def\url#1{\expandafter\string\csname #1\endcsname}

\input{custom.tex}
\usepackage{xcolor}

\usepackage{todonotes}

\title{Optimizing the ecological connectivity of landscapes with generalized flow models and preprocessing}

\author[1,2]{Fran\c{c}ois Hamonic}%
\author[2]{C\'ecile Albert}%
\author[1]{Basile Cou\"etoux}%
\author[1]{Yann Vax\`es}%
\affil[1]{Aix-Marseille Univ, CNRS, Universit\'e de Toulon, LIS, Marseille, France}
\affil[2]{Aix-Marseille Univ, CNRS, Univ Avignon, IRD, IMBE, Marseille, France}
\affil[ ]{\textit{francois.hamonic@lis-lab.fr, cecile.albert@imbe.fr, basile.couetoux@lis-lab.fr, yann.vaxes@lis-lab.fr}}

\date{}

\begin{document}

\maketitle

\begin{abstract}
  In this article we consider the problem of optimizing the landscape connectivity under a budget constraint by improving habitat areas and ecological corridors between them. We model this problem as a discrete optimization problem over graphs in which vertices represent the habitat areas and arcs represent a probability of connection between two areas that depend on the quality of the respective corridor. We propose a new generalized flow model that improves existing models for this problem. {Then, following the approach of Catanzaro et al. \cite{PREPROCESS_CANTARAZO_2011} for the robust shortest path problem, we design an improved preprocessing algorithm that reduces the size of the graphs on which we compute generalized flows.} Reported numerical experiments highlight the benefits of both contributions that allow to solve larger instances of the problem. These experiments also show that several variants of greedy algorithms perform relatively well in practice while they return arbitrary bad solutions in the worst case.
\end{abstract}


%\keywords{landscape connectivity; mixed integer programming; doubly weighted graph}

\section{Introduction}
\subsection{Context and ecological motivation}
Habitat loss is a major cause of the rapid decline of biodiversity \cite{IPBES_2019}. Further than reducing the available resources, it also  increases the discontinuities among
small habitat areas called \textit{patches},
%\textit{habitat patches}, 
%a phenomenon which is called
this is a phenomenon known under the name
\textit{habitat fragmentation} \cite{LANDSCAPE-FRAGMENTATION_JAEGER_2011}. While habitat loss tends to reduce the size of populations of animals and plants,
%, which may be deleterious to their survival in the long term, 
habitat fragmentation also makes it harder for organisms to move around in landscapes. This decreases the access to resources and the gene flow among populations.
%This can also be deleterious for their persistence in the long term. 
\textit{Landscape connectivity}, defined as the degree to which the landscape facilitates the movement of organisms between habitat patches \cite{CONNECTIVITY_TAYLOR_1993}, then becomes of major importance for biodiversity and its conservation.
Accounting for landscape connectivity in restoration or conservation plans thus appears as a key solution to maximize the return on investment of the scarce financial support devoted to biodiversity conservation.

Graph-theoretical approaches are useful in modelling habitat connectivity \cite{GRAPH_URBAN_2001}. % In this article, we consider the \textit{landscape connectivity} that is distinct from the graph theoretical notion of \textit{connectivity} defined in terms of minimum cut.
Indeed, a landscape can be viewed as a directed graph in which vertices are the habitat patches and each arc indicates a way for individuals to travel from one patch to another. The weight of a vertex represents its quality -- patch area is often used as a surrogate for quality -- and the weight of an arc mesures the difficulty for an individual to make the corresponding travel. This quantity is often approximated by a function of the border to border distance between the patches. Interestingly, this approach can be used for a variety of ecological systems like terrestrial (patches of forests in an agricultural area, networks of lakes or wetlands), riverine (patches are segments of river than can be separated by human constructions like dams that prevent fishes’ movement) or marine (patches can be reefs that are connected by flows of larvae transported by currents). With this formalism, ecologists have developed many connectivity indicators \cite{IIC_PASCUAL-HORTAL_2006, PC_SAURA_2007, REFF_MCRAE_2008} that aim to quantify the quality of a landscape with respect to the connections between its habitat patches.
\subsection{Indicators of landscape connectivity}
Among the proposed indicators, the Probability of Connectivity (PC) \cite{PC_SAURA_2007} and its derivative the Equivalent Connected Area (ECA) \cite{ECA_SAURA_2011} have received encouraging empirical support \cite{PC-SUPPORT_PEREIRA_2011,PC-SUPPORT_AWADE_2012,CONEFOR_SAURA_2009}. Given a graph $G=(V,A)$ with probability on edges $(\pi_a)_{a\in A}$ and weights on vertices $(w_v)_{v\in V}$,
$$ECA(G)=\sqrt{\sum_{s\in V} \sum_{t\in V} \left(w_s w_t \Pi_{st}\right)}$$
where $\Pi_{st}$ is the \textit{probability of connection} from patch $s$ to patch $t$
$$\Pi_{st}= \max_{P:\textnormal{$st$-path}} \prod_{a\in P} \pi_a\,.$$

{The following probabilistic analysis explains why this indicator has been called {\it Equivalent Connected Area} in \cite{ECA_SAURA_2011}. Let $\cal W$ be the area of a rectangle containing the landscape under study.} We consider a stochastic process that consists in choosing two points $p$ and $q$ uniformly at random in the rectangle. The indicator PC is the expected value of a random variable equal to $0$ if either $p$ or $q$  does not belong to a patch and $\Pi_{st}$ if $p$ belongs to $s$ and $q$ belongs to $t$ (recall that $\Pi_{st}=1$ if $s=t$). {Let $w_u$ denote the area of the patch $u.$} Since the probability that $p$ belongs to $u$ is $w_u/{\cal W}$ and the events $p \in s$ and $q\in t$ are independent, by linearity of expectation, $PC$ can be expressed as follows:
$$PC(G) = \frac{\sum_{s,t \in V} w_s w_t \Pi_{st}}{{\cal W}^2} = \frac{ECA(G)^2}{{\cal W}^2}\,.$$

% \todo{la figure est super moche}
% \begin{figure}[H]
%   \centering
%   \includegraphics[width=0.5\linewidth]{figures/figure_PC_random_process.pdf}
%   \caption{Execution times on the four case studies as a function of the budget (missing points correspond to instances that do not finish within 10 hours)}
%   \label{execution_times}
% \end{figure}

\vspace*{-0.2cm}
\begin{figure}[H]
  \centering
  \subfloat[ecological landscape \cite{LANDSCAPE-CONNECTIVITY_RUDNICK_2012}]{
    \includegraphics[scale=0.44]{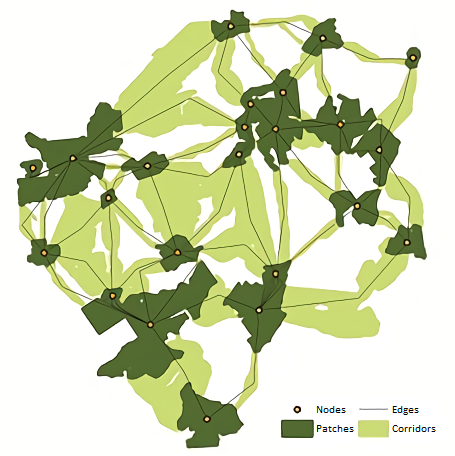}
  }
  \hspace*{1.2cm}
  \subfloat[graph representation]{
    \includegraphics[scale=0.2]{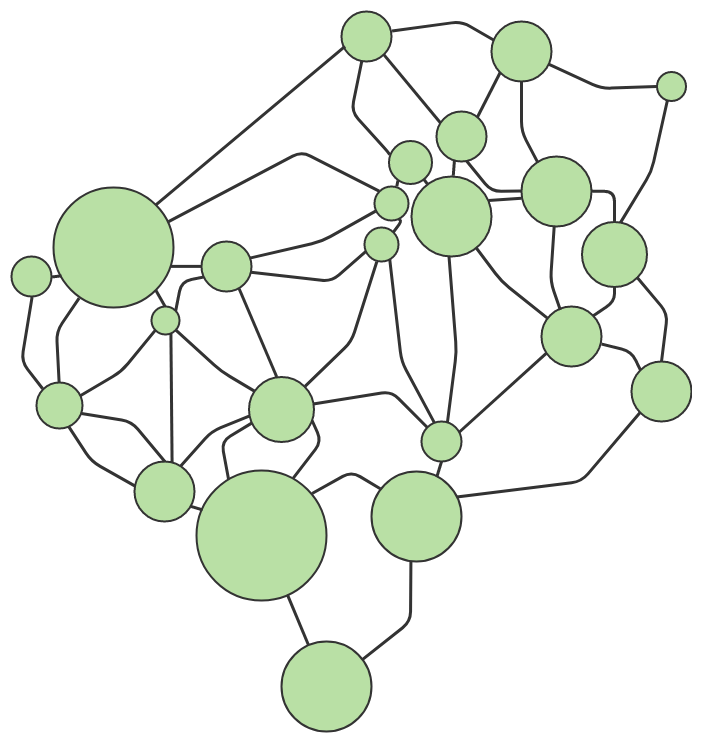}
  }
  \hspace*{0.5cm}
  \caption{Modelisation of an ecological landscape by a graph.}
\end{figure}

The equivalent connected area of a landscape is the area of a single patch whose PC value is equal the PC value of the original landscape. If the area of the patches and the landscape are normalized to make $\cal W$ equal to $1$ then PC is the square of ECA. Therefore, optimizing PC and optimizing ECA are equivalent problems. ECA is often considered by researchers interested in landscape connectivity because it represents an area, a more concrete quantity than the expected value of a random variable.
%Definitions and relationship between these two indicators are described in Section 2. 

%In this article, our purpose is not to discuss the merit of ECA with respect to other landscape connectivity indicators but rather to study the corresponding combinatorial optimization problem from an algorithmic point of view. 

% While these indicators make the use of the \textit{dispersal distance} of the studied species to assess the connectivity between pair of patches, it is highly recommended taking account to the landscape heterogeneity \cite{CONNECTIVITY-BASICS_TAYLOR_2006}. That could be done by orienting the graph and enriching the weights of vertices and arcs with population quantities, suitability of habitat patches and environment local resistance to movement. But these data are often expensive to acquire.

% The PC indicator is defined as the average probability of connection between two random points of the landscape where the probability of connection is $1$ if the two points belong to the same habitat patch, the probability of the most probable path and $0$ otherwise.

%Here, we will focus on this problem by measuring landscape connectivity with ECA. 
\subsection{{Optimizing ECA under a budget constraint}}

One key question for which landscape connectivity indicators have been used in the last years is to identify which elements of the landscape (habitat patches, corridors) should be preserved from destruction or restored in order to maintain a well-connected network of habitat under a given budget constraint \cite{FRAMEWORK_ALBERT_2017}. This translates into {identifying} the set of vertices or arcs that optimally maintains a good level of connectivity.
Many mathematical programming models have been introduced in the literature to help decision-makers to protect biodiversity. For a review of these models, we refer to the monography \cite{BILLIONNET-mono}, the survey article \cite{BILLIONNET2013514} and the references therein.
Here, we consider the budget-constrained ECA optimization problem (BC-ECA-Opt) that looks for the best combination of arcs
%to purchase, 
among a set $\Phi$
% are threatened or 
that could be restored or upgraded in order to maximize the ECA in the landscape within a limited budget:

{\bf Input:} a graph  $G=(V,A)$ with weights on vertices $(w_v)_{v\in V}$ and probabilities on the arcs $(\pi_a)_{a\in A}$, $\Phi\subseteq A$ a subset of arcs with improved probabilities $(\pi'_a)_{a\in \Phi}$ and costs $(c_a)_{a\in \Phi}$, and a budget $B\in \mathbb N$.

  {\bf Ouput:} A subset $S\subseteq \Phi$ such that $\sum_{a\in S} c_a\leq B$ maximizing $ECA(G(S))$ where $G(S)$ is the graph $G$ where we replace the value of $\pi_a$ by $\pi'_a$ for all $a\in S$.

Interestingly, this problem addresses both restoration and conservation cases. The restoration case starts from the current landscape and aims to restore a set of elements among the different feasible options. The conservation case starts with a landscape in which we search for the elements to protect among those that will be altered or destroyed if nothing is done.
%The conservation problem and the restoration problem are equivalent from an algorithmic point of view. Indeed, an instance of the conservation problem can be viewed as an instance of the restoration problem in which the initial landscape is the landscape in which all threatened elements are degraded and the set $\Phi$ of elements than can be restored in the instance of the restoration problem is the set of elements that can be protected in the instance of the conservation problem.  

\subsection{Previous works}

Until now, ecologists have mostly tackled this problem by ranking each conservation or restoration option by its independent contribution to ECA, i.e. the amount by which ECA varies if the option is purchased alone. Such an approach overlooks the cumulative effects of the decisions made like unnecessary redundancies or potential synergistic effects, e.g. {improving two consecutive corridors results in a greater increase in ECA than the sum of the increases achieved by improving each corridor independently}. This could lead to solutions that are more expensive or less beneficial to ECA than an optimal solution. {%Section \ref{section_greedy_bad_cases} presents small instances in which such bad behaviours occur.} 
    Some studies have tried to overcome this limitation by considering tuples of options \cite{MULTINODE_RUBIO_2015, MULTINODE_PEREIRA_2017}. In \cite{MULTINODE_RUBIO_2015}, the authors show that the brute force approach rapidly becomes impractical for landscape with more than 20 patches of habitat.
    Few studies have explored the search for an optimal solution but, in most cases, the underlying graph was acyclic (river dendritic networks).  For instance,  a polynomial time approximation scheme has been proposed when the underlying graph is a tree. Indeed, \cite{APPROX-TREES_WU_2014} describes a dynamical programming algorithm with rounding that computes,  for any $\epsilon,$ a $(1-\epsilon)$-approximated solution in time $n^8/\epsilon$ where $n$ is the number of nodes of the tree.}
More recently, \cite{PLNE_XUE_2017} has introduced a XOR sampling method based on a mixed integer formulation. To our knowledge, this is the unique linear programming formulation of BC-ECA-Opt already proposed. Solving optimally the problem with their mixed integer formulation does not scale to landscapes with few hundreds patches. %{However, their main contribution is not this MIP formulation but a XOR sampling method for finding approximate solutions of larger instances.}

%\subsection{Relation with shortest path computation}
To compute $ECA(G)$ we have to compute $\Pi_{st}$ for every pair of vertices $s,t\in V.$ If we consider the length function $l$ defined by $l_a = -\log \pi_a$ for $a\in A$ then the distance function $d$ induced by $l$ on $G$ verifies $d(s,t)=-\log(\Pi_{st})$ for $s,t\in V$. Thus our problem is closely related to solving shortest path problems on a graph with different length functions, namely one length function for each restoration plan (or {\it scenario}).  A preprocessing step has been proposed by Catanzaro et al. \cite{PREPROCESS_CANTARAZO_2011} to address such problems. It consists in identifying a subset of arcs that can either be removed or contracted to reduce the size of the graph considered. More formally, following \cite{PREPROCESS_CANTARAZO_2011}, we call an arc $(u,v)$ $t$-strong if it belongs to a shortest path from $u$ to $t$ in every possible scenario. Symetrically, we call an arc $(u,v)$ $t$-useless if there is no scenario such that $(u,v)$ is on a shortest path from $u$ to $t.$ Useless arcs were called $0$-persistent in \cite{PREPROCESS_CANTARAZO_2011}.  Since we need to specify a target vertex $t$ for which $(u,v)$ is useless, we didn't adopt the same terminology. In \cite{PREPROCESS_CANTARAZO_2011}, given an arc $(u,v)$ and a vertex $t$, the authors give a sufficient  (but not necessary) condition for $(u,v)$ to be $t$-strong and use it to design a $O(|A| + |V| \log |V|)$ algorithm that decides, in most cases, whether $(u,v)$ is $t$-strong. The authors also propose a $O(|A| + |V| \log |V|)$ algorithm to test whether a given arc $(u,v)$ is $t$-useless.

\subsection{Contribution}
The contribution of this article is twofold.
First, we propose a new mixed integer linear formulation which is more compact than the one of \cite{PLNE_XUE_2017}. Then, we design a new preprocessing algorithm that computes $t$-strong and $t$-useless arcs for all $t\in V.$

%Therefore, we did not perform numerical experiments to compare our methodology against those developed in \cite{PLNE_XUE_2017}.
%Here, our aim is to propose a more compact mixed integer formulation than the one proposed in \cite{PLNE_XUE_2017} for the budget-constrained ECA optimization problem. Since he XOR sampling method can be applied on our more compact mix integer formulation, we did not perform numerical experimets to compare our methodology with those developed in \cite{PLNE_XUE_2017}.

Our new mixed integer formulation is based on a generalized flow formulation (see for instance \cite{GENERALIZED-FLOW_AHUJA_1988}) instead of a standard network flow formulation. This leads to two improvements. Firstly, our formulation has a linear objective function as opposed to {the model of \cite{PLNE_XUE_2017}} that was using a piecewise constant approximation and additional binary variables to handle non-linearity.
Secondly, the new formulation aggregates into a single generalized flow the contribution to the connectivity of several source/sink pairs having the same source whereas the previous model treated every pair separately.
  {More precisely, our model uses $O(|V||A|)$ flow variables with $O(|V|^2)$ constraints, and $|\Phi|$ integer variables, while  the model of \cite{PLNE_XUE_2017} was using of $O(|V|^2|A|)$ flow variables with $O(|V|^3)$ constraints, and $O(|\Phi|+|V|^2)$ integer variables.} Recall that $\Phi$ is the set of elements that are threatened or could be restored.
  %Note that this model could be used as well for\cite{PLNE_XUE_2017} XOR sampling method.
  %On the other side, the new formulation allows to aggregate into one generalized flow the contributions of all the pairs $ut$ having the same sink $t$. 
  %This formulation could also easily be extended to model the restoration or conservation of habitat patches.

  {The other contribution of this article is a preprocessing step that speeds-up the resolution of the problem by reducing the size of the directed graph on which a generalized flow to a particular sink $t$ have to be computed.
    %The previous known algorithm from \cite{PREPROCESS_CANTARAZO_2011} computed a subset of the $t$-strong arcs. 
    Our algorithm improves the one of \cite{PREPROCESS_CANTARAZO_2011} in two ways. First, we replace the sufficient condition defined in \cite{PREPROCESS_CANTARAZO_2011} by a necessary and sufficient condition. This allows us to compute the entire set of $t$-strong arcs instead of a subset of them. We also improve the time complexity by a factor $|V|$ as we only need to run our algorithm on each arc instead of each pair of arc and vertex.}

\subsection{Organization of the article}

%Moreover its complexity is $|V|$ times faster than the previous algorithm.
%In Section \ref{section_problem_definition} we give more details about the landscape modelization and describe formally the problem. 
{Section \ref{section_new_model} is devoted to the description of our mixed integer formulation of the problem.
  In Section \ref{preprocessing} we first explain why the removal of $t$-useless arcs and the contraction of $t$-strong arcs does not modify the objective function of any restoration/conservation plan. Then, we design and analyze a $O(|A| + |V| \log |V|)$ algorithm that computes the set of all vertices $t$ for which a given arc $(u,v)$ is $t$-strong or the set of all vertices $t$ for which $(u,v)$ is $t$-useless.
  In section \ref{section_greedy_bad_cases}, we describe several greedy algorithms for BC-ECA-Opt and provide some instances where they compute solutions far from being optimal.
  Section \ref{experiments} compares our optimization approach to simple greedy algorithms in terms of running times and quality of the solutions found on a set of experimental cases. These experiments show that the preprocessing is very effective and that the greedy approach performs quite well on these instances.}

\section{An improved MIP formulation for BC-ECA-Opt\label{section_new_model}}

We decompose $ECA(G)$ as $\sum_{t\in V} w_t f_t$ where $f^{t}=\sum_{s \in V} w_s \cdot \Pi_{st}$.
We first show that $f^{t}$ can be expressed as the maximum quantity of generalized flow that can be sent to $t$ across the network if, for every patch $s \in V$, $w_s$ units of flow are available at $s$ and appropriate multipliers are chosen for each arc of the network. Recall that a generalized flow differs from a standard flow by the fact that each arc $a$ has a multiplier $\pi_a$ such that the quantity of flow leaving  arc $a$ is equal to the quantity of flow $\phi_a$ entering in $a$ multiplied by $\pi_a$ (see \cite{GENERALIZED-FLOW_AHUJA_1988} for an introduction to network flows).
%In our case $\pi_a=\exp(-\alpha l_a)$. 
For each vertex $u\in V,$ let $\delta^{\text{out}}_u$ be the set of arcs leaving $u$ and $\delta^{\text{in}}_u$  the set of arcs entering $u$.
As explained below, the linear program \hyperref[opt-P]{$\mathcal{(P)}$} has $f^t$ as optimum value.

\begin{optimisation_program}
  \name{\label{opt-P}{$\mathcal{(P)}$}}
  \sense{max}
  \objective{$z$}
  \constraints{
    \constraint{$\sum\limits_{a \in \delta^{\text{out}}_u} \phi_a - \sum\limits_{b \in \delta^{\text{in}}_u} \pi_b\cdot  \phi_b \leq w_u $} {$u \in V\setminus \{t\}$}
    \constraint{$\sum\limits_{a \in \delta^{\text{out}}_t} \phi_a - \sum\limits_{b \in \delta^{\text{in}}_t} \pi_b\cdot \phi_b = w_t - z$} {}
    \constraint{$\phi_a\ge 0$} {$a \in A$}
  }
\end{optimisation_program}

Constraints \hyperref[opt-P]{(A1)} require that the total quantity of flow leaving $u$ is at most $w_u$ plus the total quantify of flow entering $u$, i.e. the quantity of flow available in vertex $u$. Constraint \hyperref[opt-P]{(A2)} requires that $z$ is equal to the total quantity of flow entering $t$ plus $w_t$ minus the total quantity of flow leaving $t$. Finally, constraints \hyperref[opt-P]{(A3)} state that each arc carries a non negative quantity of flow from its source to its sink.

\begin{lemma}\label{flowOnShortestPaths}
  Any optimal solution of $\mathcal{(P)}$ is obtained by sending, for every vertex $s\in V-\{t\},$ $w_s$ units of flow from $s$ to $t$ along most reliable $st$-paths.
\end{lemma}
\begin{proof}
  First notice that the quantity of flow arriving at $t$ when $w_s$ units flow are sent from $s$ to $t$ along an $st$-path $P$ is $w_s$ times the probability of path $P.$ Let $\phi'$ be an optimal solution of $\mathcal{(P)}$ maximizing the quantity of flow routed along a path which is not a most reliable path. Suppose, by contradiction, that $\phi'$ sends $\epsilon>0$ units of flow along an $st$-path $P'$ whose probability is smaller than the probability of a most reliable $st$-path $P.$ Let $\phi$ be the flow obtained from $\phi'$ by decreasing the flow sent on path $P'$ by $\epsilon$ and increasing the flow sent on path $P$ by $\epsilon.$ Since the probability of $P$ is larger than the probability of $P',$ $\phi$ sends more flow to $t$ than $\phi',$ leading to a contradiction with the choice of $\phi'.$
\end{proof}

\begin{corollary}\label{ft-definition}
  The optimal value of $\mathcal{(P)}$ is $f^t = \sum_{s \in V} w_s \cdot \Pi_{st}$.
\end{corollary}
\begin{proof}
  Since the objective is to maximize $z,$ no flow leaves $t$ in any optimal solution of $\mathcal{(P)},$ i.e. $\sum\limits_{a \in \delta^{\text{out}}_t} \phi_a=0.$  Constraint \hyperref[opt-P]{(A2)} ensures that $z$ is the quantity of flow received by $t$ plus $w_t.$ By Lemma \ref{flowOnShortestPaths}, there exits an optimal solution $\phi$ such that every vertex $s$ distinct from $t$ send $w_s$ units of flow on a most reliable $st$-path. Hence, for every vertex $s$ distinct from $t$, the flow received by $t$ from $s$ is $w_s\cdot \Pi_{st}$ and thus the value of $z$ is $\sum_{s \in V} w_s \cdot \Pi_{st}.$
\end{proof}

With this linear programming definition of $f^t,$ we can now present our MIP program of BC-ECA-Opt. For each arc $a=(u,v) \in \Phi,$ we add another arc $a'=(u,v)$ of probability $\pi_{a'}=\pi'_a$ that can be viewed as an improved copy of arc $a$. This improved copy can be used only if the improvement of arc $a$ is purchased. For each arc $a\in \Phi,$ $x_a$ is equal to one if the improvement of arc $a$ is purchased  and zero otherwise. We denote by $\Psi$ the set of improved copies of arcs in $\Phi.$   In the following MIP program, $\delta^{\text{out}}_u$ and $\delta^{\text{in}}_u$ are defined with respect to the set of arcs $A'=A\cup \Psi.$

\begin{optimisation_program}
  \name{BC-ECA-Opt\label{bc-eca-opt}}
  \sense{max}
  \objective{$\sum\limits_{t\in V} w_t \cdot f^t$}
  \constraints{
    \constraint{$\sum\limits_{a \in \delta^{\text{out}}_u} \phi^t_a - \sum\limits_{b \in \delta^{\text{in}}_u} \pi_b\cdot  \phi^t_b \leq w_u $} {$t\in V$, $u \in V\setminus \{t\}$}
    \constraint{$\sum\limits_{a \in \delta^{\text{out}}_t} \phi^t_a - \sum\limits_{b \in \delta^{\text{in}}_t} \pi_b\cdot  \phi^t_b =  w_t - f_t$} {$t\in V$}
    \constraint{$\phi^{t}_{a'} \le x_a\cdot M_a$} {$t\in V$, $a \in \Phi$}
    \constraint{$\sum\limits_{a \in F} c_a \cdot x_a \le B $} {}
    \constraint{$x_a\in \{0,1\}$} {$a \in \Phi$}
    \constraint{$\phi^{t}_a\geq 0$} {$t \in V$, $a \in \Phi$}
  }
\end{optimisation_program}

Constraints of \hyperref[bc-eca-opt]{(B1)} and \hyperref[bc-eca-opt]{(B2)} are simply constraints of \hyperref[opt-P]{(A1)} and \hyperref[opt-P]{(A2)} for all possible target vertex $t$. For each arc $a\in \Phi,$ the big-M constraints \hyperref[bc-eca-opt]{(B3)} ensure that if the improvement of arc $a$ is not purchased, i.e. $x_a = 0,$ then the flow on arc $a'$ is null. The constant $M_a$ is an upper bound of the flow value on arc $a.$ We could simply take $M_a = \sum_{u\in V} w_u$ for all $a\in \Phi$ but more precise estimations are possible for a better linear relaxation and thus a faster resolution of the MIP program. Constraint \hyperref[bc-eca-opt]{(B4)} ensures that the total cost of the improvements is at most $B.$ This MIP program has $O(|V|(|A|+|\Phi|))$ flow variables, $|\Phi|$ binary variables and $O(|V|^2+|V||\Phi|)$ constraints.

\subsection{Extension to patch improvements} \label{mip_extensions}

{Here, we explain how to extend our model to the version with patch improvements where it is possible to increase the weight of a vertex $u$ from $w_u$ to $w^+_u$ at cost $c_u.$} For that, we process all the occurrences of $w_u$ as follows. Let $y_u$ be a binary variable equal to $1$ if the vertex $u$ is improved and $0$ otherwise. We add a term $y_u c_u$ in the budget constraint for every vertex $u$ in the set $W$ of vertices that can be improved. When $w_u$ appears as an additive constant, we simply replace it by $w_u + y_u (w^+_u-w_u)$. Note that $w_u$ only appears as a coefficient in the objective function in the form $w_u f_t$. In this case, to avoid a quadratic term, we use a standard McCormick linearization \cite{DBLP:journals/mp/McCormick76}. We replace the product $w_u f_t$ by $w_u f_t + (w^+_u-w_u) f_t'$ where $f_t'$ is a new variable that is equal to $f_t$ if $y_u=1$ and $0$ otherwise. To achieve this values of $f_t'$, for all $u\in V,$ we add the constraints $f_t' \le f_t$ and $f_t' \le y_u M$ where $M$ is larger than any values of $f_t$. As it is a maximization program and $f_t'$ appears with a positive coefficient in the objective function, the first constraint guaranties that $f_t' = f_t$ if $y_u = 1$ in any optimal solution and the second constraint guaranties that $f_t' = 0$ if $y_u = 0$.

\section{Preprocessing \label{preprocessing}}

%For that, we introduce a notion of \textit{strongness} of an arc with respect to a target vertex $t \in V$.  An arc $(u,v)$ is said to be \textit{$t$-strong} if, for every scenario $x \in \{0,1\}^\Phi$, arc $(u,v)$ belongs to a shortest path from $u$ to $t,$ i.e. $d_x(u,t) = d_x(u,v) + d_x(u,t)$. If arc $(u,v)$ belongs to every shortest $ut$-path for every scenario $x \in \{0,1\}^\Phi,$ then we will say that it is strictly $t$-strong. Finally, note that the notion of stongness can be extended to a subset of arc with the same source.  We will say that a subset $\Gamma \subseteq \Gamma_u^+$ of arcs is $t$-strong (resp. strictly $t$-strong) if, for every scenario $x \in \{0,1\}^\Phi,$ a shortest $ut$-path (resp. every shortest $ut$-path) intersects $\Gamma.$ An arc $(u,v)$ is said to be \textit{$t$-useless} if, for every scenario $x \in \{0,1\}^\Phi$, arc $(u,v)$ does not belong to any shortest $ut$-path when the lengths of the arcs are set according to the scenario $x$.
The size of the mixed integer programming formulation of BC-ECA-Opt given in Section \ref{section_new_model} grows quadratically with the size of the graph that represents the landscape. In this section, we describe a preprocessing step that reduces the size of this graph.
{For that, we adopt the approach used by Catanzaro et al. \cite{PREPROCESS_CANTARAZO_2011} for the robust shortest path problem.
%has been proposed by Catanzaro et al. \cite{PREPROCESS_CANTARAZO_2011} to address such problems. It consists in identifying a subset of arcs that can either be removed or contracted to reduce the size of the graph considered. 
We introduce a notion of \textit{strongness} of an arc with respect to a target vertex $t \in V$. We call a solution $x\in \{0,1\}^{\Phi}$ of BC-ECA-Opt a {\it scenario} and say that the distances are {\it computed under scenario $x$} when the length of every $a\in A$ is $l^-_a$ if $x_a=1$ and $l_a$ otherwise. We denote $d_x(s,t)$ the distance between $s$ and $t$ when the arc lengths are set according to the scenario $x.$
%i.e. the length of $a$ is $l^-_a$ if $x_a=1$ and $l_a$ otherwise.  
Following \cite{PREPROCESS_CANTARAZO_2011}, an arc $(u,v)$ is said to be \textit{$t$-strong} if, for every scenario $x \in \{0,1\}^\Phi$, $(u,v)$ belongs to a shortest path from $u$ to $t,$ i.e. $d_x(u,t) = d_x(u,v) + d_x(v,t)$ for every $x\in\{0,1\}^\Phi$. An arc $(u,v)$ is said to be \textit{$t$-useless} if, for every scenario $x \in \{0,1\}^\Phi$, arc $(u,v)$ does not belong to any shortest $ut$-path when arc lengths are set according to the scenario $x.$ Useless arcs were called $0$-persistent in \cite{PREPROCESS_CANTARAZO_2011} but here we need to specify the target $t$. We denote by $S(t)$ the set of arcs $(u,v)$ such that $(u,v)$ is $t$-strong and by $W(t)$ the set of arcs $(u,v)$ such that $(u,v)$ is $t$-useless.}
%The problem of identifying $S(t)$ and $W(t)$ has been studied as a preprocessing step for other combinatorial optimization problems having as input a graph whose arc lengths may change \cite{PREPROCESS_CANTARAZO_2011}. 
In \cite{PREPROCESS_CANTARAZO_2011}, given an arc $(u,v)$ and a vertex $t$, the authors identify a sufficient  (but not necessary) condition for $(u,v)\in S(t)$ and use it to design a $O(|A| + |V| \log |V|)$ algorithm that can, in most cases, identify if $(u,v)\in S(t)$. The authors also propose a $O(|A| + |V| \log |V|)$ algorithm to test whether a given arc $(u,v)$ belongs to $W(s)$ or not.
%Of course, the naive algorithm that tests whether $(u,v)$ belongs to a shortest  $ut$-path for every scenario $x\in \{0,1\}^{\Phi}$ is very inefficient. 
In Section \ref{ComputeSt}, given an arc $(u,v)$ of $G$, we show how to adapt the Dijkstra's algorithm to compute in $O(|A| + |V| \log |V|)$ the set of all vertices $t$ such that $(u,v) \in S(t)$ or the set of all vertices $t$ such that $(u,v) \in W(t)$. This improves the results of \cite{PREPROCESS_CANTARAZO_2011} by providing a necessary and sufficient condition for $(u,v)$ to be $t$-strong, i.e. we compute the entire set $S(t)$ while the algorithm of \cite{PREPROCESS_CANTARAZO_2011} computes a subset of $S(t).$ It also reduces the time complexity for computing $S(t)$ and $W(t)$ for all $t$ by a factor $|V|$ as we only need to run the algorithms on each arc instead of each pair of arc and vertex.
%we describe and analyse an algorithm that computes the set $S(t)$ and $W(t)$ for every sink $t$ in a more efficient way. 
Then, we explain how the knowledge of $S(t)$ and $W(t)$ for all $t$ can be used to define a smaller equivalent instance of the problem.

\subsection{Computing $S(t)$ for all $t$}\label{ComputeSt}

Given a scenario $x\in \{0,1\}^\Phi,$ the fiber $F_x(u,v)$ of arc $(u,v)\in A$ is the set of vertex $t$ such that $(u,v)$ belongs to a shortest path from $u$ to $t$ when arc lengths are set according to $x,$ i.e.

$$F_x(u,v)=\{t\in V : d_x(u,t)= l_x(u,v)+d_x(v,t)\}$$

Let $F(u,v)$ be the intersection of the fibers of $(u,v)\in A$ over all possible scenarios $x\in \{0,1\}^\Phi,$ i.e. $F(u,v):=\bigcap_{x} F_x(u,v).$ By definition, an arc $(u,v)$ is $t$-strong if $t$ belongs to $F_x(u,v)$ for every scenario $x$ i.e.
$$S(t)=\left\{(u,v) : t\in F(u,v)\right\}$$
In order to compute every $S(t)$, we first compute $F(u,v)$ for every arc $(u,v)\in A$ and then we transpose the representation to get $S(t)$ for every vertex $t\in V$.
Let $y\in \{0,1\}^\Phi$ be the following scenario :
\begin{equation}\label{scenario-definition}
  y_{wt}=\left\{\begin{array}{ll} 0 & w\in  F(u,v) \textbf{or} (w,t)=(u,v)\\ 1 &  w\notin  F(u,v) \end{array} \right.
\end{equation}
\begin{lemma}\label{F-inter}
  The intersection $F(u,v)$ of the fibers of $(u,v)\in A$ over all possible scenarios is the fiber of $(u,v)$ under the scenario $y,$ i.e. $F(u,v) = F_y(u,v).$
\end{lemma}
\begin{proof}
  Since the inclusion $F(u,v)\subseteq F_y(u,v)$ is obvious, it suffices to prove that $F_y(u,v)\subseteq F_x(u,v)$ for any scenario $x\in \{0,1\}^\Phi.$ By way of contradiction, suppose that $x$ is a scenario such that $F_y(u,v) \setminus F_x(u,v)$ contains a vertex $t$ at minimum distance from $u$ in the scenario $y.$ Let $P$ be a shortest $ut$-path containing $(u,v)$ under the scenario $y.$ For every vertex $w$ of $P,$ $d_y(u,w)\le d_y(u,t)$ and $w\in F_y(u,v).$ Hence, by the choice of $t$, $w$ belongs to $F_x(u,v)$ for every scenario $x.$ Therefore $w$ belongs to $F(u,v)$ and the length of every arc of $P$ is set to its upper bound in the scenario $y,$ i.e. $l_y(P)=l(P).$ Now, let $Q$ be a shortest $ut$-path in the scenario $x.$ If the path $Q$ contains a vertex $z \in F(u,v)$ then $d_x(u,t)=d_x(u,z)+d_x(z,t)=l_x(u,v)+d_x(v,z)+d_x(z,t)= l_x(u,v)+d_x(v,t),$ a contradiction with $t\notin F_x(u,v).$  Therefore, the vertices of $Q$ do not belong to $F(u,v)$ and thus their lengths are set to their lower bound in $y,$ i.e. $l(Q)=l^-(Q).$ We deduce that $l_x(P) \le l_y(P) \le l_y(Q) \le l_x(Q)$, which contradicts $t\notin F_x(u,v).$
\end{proof}

We describe a $O(|A| + |V| \log |V|)$ time algorithm that, given an arc $(u,v)$, computes simultaneously the scenario $y$ defined by \eqref{scenario-definition} and the fiber $F_y(u,v)$ of arc $(u,v)$ with respect to this scenario. The algorithm is an adaptation of the Dijkstra's shortest path algorithm that assigns colors to vertices. We prove that it colors a vertex $w$ in blue if $w$ belongs to $F_y(u,v)$ and in red otherwise. At each step, we consider a subset of vertices $S\subseteq V$ whose colors have been already computed.
Before the first iteration, $S=\{u\}$ and the length of every arc leaving $u$ is set to its lower bound except the length of $(u,v)$ which is set to its upper bound according to scenario $y.$  At each step, since the color of every vertex of $S$ is known, the length, under scenario $y,$ of every arc $(w,t)$ with $w\in S$ is also known. Therefore, it is possible to find a vertex $t\in V-S$ at minimum distance from $u,$ under scenario $y,$ in the subgraph $G[S\cup\{t\}]$ induced by $S\cup\{t\}.$ Following Dijkstra's algorithm analysis, we know that the distance under scenario $y$ from $u$ to $t$ in $G[S\cup\{t\}]$ is in fact the distance between $u$ and $t$ in $G$. This allows us to determine if there exists a shortest path from $u$ to $t$ in the scenario $y$ passing via $(u,v)$ and to color the vertex $t$ accordingly.  We end-up the iteration with a new vertex $t$ whose color is known and that can be added to $S$ before starting the next iteration. The algorithm terminates when $S=V.$

\medskip

\begin{algorithm}[H]
  \caption{Computes the set $S$ of vertices $t$ such that $(u,v)$ is $t$-strong\label{algo-strong}}
  \SetKwInOut{Input}{Input}
  \SetKwInOut{Output}{Output}
  \Input{$G=(V, A, l, l^-)$, $(u,v) \in A$}
  \Output{$\{t \in V : (u,v) \text{is} t\text{-strong}\}$}
  \ForEach{$(u,w) \in \delta^{\text{out}}_u \setminus \{(u,v)\}$} {
  $d(w) \gets l^-_{uw}$\\
  $\gamma(w) \gets red$
  }
  $d(v) \gets l_{uv}$;
  $\gamma(v) \gets blue$\\
  $S \gets \{u\}$;
  $\gamma(u)\gets red$\\

  \While{$S \neq V$} {
  Pick $t \in V-S$ with smallest $d(t)$ breaking tie by choosing a vertex $t$ such that $\gamma(t)=blue$ if it exists\\
  \If{$\gamma(t)$ is $blue$}{
  \ForEach{$(t,w) \in \delta^{\text{out}}_t$ such that $d(w) \ge d(t)+l_{tw}$} {
  $d(w) \gets d(t)+l_{tw}$\\
  $\gamma(w) \gets blue$
  }
  }
  \Else{
  \ForEach{$(t,w) \in \delta^{\text{out}}_t$ such that $d(w) > d(t)+l^-_{tw}$} {
  $d(w) \gets d(t)+l^-_{tw}$\\
  $\gamma(w) \gets red$
  }
  }
  $S \gets S \cup \{t\}$
  }
  \Return $\{t\in V : \gamma(t)=blue\}$
\end{algorithm}

For every vertex $w\in V-S$,  the estimated distance $d(w)$ is the length of a $uw$-path under the scenario $y$ in the subgraph $G[S\cup \{w\}].$ An arc $a$ is blue if $a=uv$ or if its origin is blue, the other arcs are red, i.e. $a$ is blue if $y_a=0$ and red if $y_a=1.$ The estimated color $\gamma(w)$ of $w$ is blue if there exists a blue $uw$-path of length $d(w)$ in $G[S\cup \{w\}]$ and red otherwise.  The correctness of Algorithm \ref{algo-strong} follows from the following Lemma.
\begin{lemma}\label{correctness-Algo1}
  For every vertex $t\in S, $ $\gamma(t)$ is blue if and only if $t\in F_y(u,v).$
\end{lemma}
\begin{proof}
  We proceed by induction on the number of vertices of $S.$ When $S=\{u\},$ the property is verified. Now, suppose the property true before the insertion in $S$ of the vertex $t$ such that $d(t)$ is minimum. By induction hypothesis, the lengths of arcs having their source in $S$ are set according to $y.$ Therefore, following Dijkstra's algorithm analysis, we deduce that $d(t)$ is the length of the shortest path from $u$ to $t$ in the graph $G$ under scenario $y.$ If $t$ has been colored blue then there exists a blue vertex $w\in S$ such that $d(t)=d(w)+l_{wt}.$ By induction hypothesis $w\in F_y(u,v)$ and there exists a shortest $ut$-path under scenario $y$ containing $(u,v),$ i.e. $t\in F_y(u,v).$ Now, suppose that $t$ has been colored in red. By contradiction, assume there exists a shortest $ut$-path under scenario $y$ that contains $(u,v).$ This path cannot contain a vertex outside $S$ except $t$ because otherwise, since arc lengths are non-negative, its length according to $y$ would be greater than $d(t)$ by the choice of $t.$ Therefore, the predecessor $w$ of $t$ in this path belongs to $S.$ Since $w$ belongs to a shortest $ut$-path passing via $(u,v),$ by induction, it was colored blue. But in this case, there exists a blue vertex $w$ such that $d(w)+l(w,t)=d(t),$ and $t$ was colored blue as well, a contradiction.
\end{proof}

We are now ready to state the main result of this section.

\begin{proposition}
  Given a graph $G=(V,A),$ two arc-length functions $l^-$ and $l$ such that $0 \le l^-_a \le l_a$ for every arc $a\in A$, and an arc $(u,v)$, Algorithm \ref{algo-strong} computes in $O(|A| + |V| \log |V|)$ the set of vertex $t$ such that $(u,v)$ is $t$-strong.
\end{proposition}

\begin{proof}
  The lemma \ref{F-inter} shows that given an arc $(u,v)$ the set of vertices $S$ such that $t\in S$ if and only if $(u,v)$ is $t$-strong is a fiber for a specific scenario. The lemma \ref{correctness-Algo1} shows that the algorithm compute this fiber. Therefore the Algorithm \ref{algo-strong} is correct. Analogously to the Dijkstra's algorithm, Algorithm \ref{algo-strong} can be implemented to run in $O(|A|+|V|\log |V|)$ by using a Fibonacci heap as priority queue.
\end{proof}

\medskip

\subsection{Computing $W(t)$ for all $t$}

An adaptation of Algorithm \ref{algo-strong} can compute, given an arc $(u,v)$ the set of vertices $W$ such that $t\in W$ if and only if $(u,v)$ is $t$-usless. Before describing this adaptation, we explain how the two problems  are related. For that, we first introduced a strengthening of the notion of strongness. We will say that an arc $(u,v)$ is strictly $t$-strong if it belongs to all shortest $ut$-paths for every scenario $x \in \{0,1\}^{\Phi}.$ Recall that $t$-strongness requires only the existence of a shortest $ut$-path passing via $(u,v)$ for every scenario $x \in \{0,1\}^{\Phi}.$ Algorithm \ref{algo-strong} can be easily adapted to compute for every arc $(u,v)$ the set of vertex $t$ such that $(u,v)$ is strictly $t$-strong. It suffices to change the way, the algorithms breaks tie between a red and a blue path and the choice of $t$ in case of tie. Namely, in the first internal loop, the condition for coloring $w$ in blue becomes $d(w)>d(t)+l_{tw}$ while the condition for coloring $w$ in red in the second internal loop becomes $d(w)\ge d(t)+l^-_{tw}.$ Moreover, when we choose $t$ such that $d(t)$ is minimal, we break tie by choosing a red vertex if it exists. Clearly, these small changes exclude the existence of a red path of length $d(w)$ between $u$ and a blue vertex $w$. Therefore, $(u,v)$ belongs to every path of length $d(w)$ and $(u,v)$ is strictly $w$-strong whenever $w$ is blue. We call the resulting algorithm the strict version of Algorithm \ref{algo-strong}. The next step is to extend the notion of strict strongness to a subset of arcs having the same source.  For any vertex $u\in V,$ a subset $\Gamma \subseteq \Gamma_u^+$ of arcs is strictly $t$-strong if, for every scenario $x \in \{0,1\}^\Phi,$ all shortest $ut$-paths intersect $\Gamma.$ By definition, an arc $(u,v)$ is $t$-useless if and only if $\Gamma_u^+\setminus\{(u,v)\}$ is strictly $t$-strong. Indeed, every $ut$-path avoiding $(u,v)$ intersects $\Gamma_u^+\setminus\{(u,v)\}$ and, conversely, $(u,v)$ belongs to every $ut$-path avoiding $\Gamma_u^+\setminus\{(u,v)\}$. Hence, computing the set of vertex $t$ such that $(u,v)$ is $t$-useless amounts to compute the set of vertex $t$ such that
$\Gamma_u^+\setminus\{(u,v)\}$ is strictly $t$-strong. An algorithm that computes this set of vertices can be obtained from the strict version of Algorithm \ref{algo-strong} by modifying only the initialization step: arc $(u,v)$ is colored in red and its length is set to $l^-_{uv}$ while arcs of $\delta^{\text{out}}_u \setminus \{(u,v)\}$ are colored in blue and their lengths are set to their upper bound. A correctness proof very similar to the one of Algorithm \ref{algo-strong} (and that we will not repeat) shows that a vertex is colored blue by Algorithm \ref{algo-useless} if and only if $(u,v)$ is $t$-useless. Since the two algorithms have clearly the same time complexity, we conclude this section with the following result.

\medskip

\begin{algorithm}[t]
  \caption{Computes the set of vertex $t$ such that $(u,v)$ is $t$-useless\label{algo-useless}}
  \SetKwInOut{Input}{Input}
  \SetKwInOut{Output}{Output}
  \Input{$G=(V, A, l, l^-)$, $(u,v) \in A$}
  \Output{$\{t \in V : (u,v) \text{is} t\text{-useless}\}$}
  $d(v) \gets l^-_{uv}$;
  $\gamma(v) \gets red$\\
  \ForEach{$(u,w) \in \delta^{\text{out}}_u \setminus \{(u,v)\}$} {
    $d(w) \gets l_{uw}$\\
    $\gamma(w) \gets blue$
  }
  $S \gets \{u\}$;
  $\gamma(u)\gets blue$\\

  \While{$S \neq V$} {
  Pick $t \in V-S$ with smallest $d(t)$ breaking tie by choosing a vertex $t$ such that $\gamma(t)=red$ if it exists\\
  \If{$\gamma(t)$ is $blue$}{
  \ForEach{$(t,w) \in \delta^{\text{out}}_t$ such that $d(w) > d(t)+l_{tw}$} {
  $d(w) \gets d(t)+l_{tw}$\\
  $\gamma(w) \gets blue$
  }
  }
  \Else{
  \ForEach{$(t,w) \in \delta^{\text{out}}_t$ such that $d(w) \ge d(t)+l^-_{tw}$} {
  $d(w) \gets d(t)+l^-_{tw}$\\
  $\gamma(w) \gets red$
  }
  }
  $S \gets S \cup \{t\}$
  }
  \Return $\{t\in V : \gamma(t)=blue\}$
\end{algorithm}

\begin{proposition}
  Given a graph $G=(V,A),$ two arc-length functions $l^-$ and $l$ such that $0 \le l^-_a \le l_a$ for every arc $a\in A$, and an arc $(u,v)$, Algorithm \ref{algo-useless} computes in $O(|A| + |V| \log |V|)$ the set of vertex $t$ such that $(u,v)$ is $t$-useless.
\end{proposition}

\subsection{Operations to reduce the size of the graph}

\noindent
\textbf{Removal of an arc from $W(t)$.} Let $(u,v) \in W(t)$, since for all scenarios $x\in \{0,1\}^{\Phi}$, $(u,v)$ does not belong to any shortest path from $u$ to $t,$ its removal does not affect the distance from any vertex to $t$. It is clear, from the definition of ECA, that the removal of $(u,v)$ does not affect the contribution of $t$ to ECA. Let $f^x_t(G):=\sum_{s\in V} w_s w_t \Pi^x_{st}$ be the contribution to ECA of all the pairs having sink $t$ in the graph $G$ when the probability of connection $\Pi^x_{st}$ is computed under the scenario $x.$

\begin{lemma}\label{remove}
  Let $(u,v)\in W(t)$ and let $G'$ be a graph obtained from $G$ by removing arc $(u,v).$ Then, for all scenario $x\in \{0,1\}^\Phi,$ it holds that $f^x_t(G)=f^x_t(G').$
\end{lemma}

\begin{proof}
  For every scenario $x\in \{0,1\}^{\Phi}$, $(u,v)$ does not belong to any shortest path from $u$ to $t.$ Therefore, the removal of $(u,v)$ cannot affect the probability of connection $\Pi_{ut}$ for any vertex $t$, Thus $f^x_t(G)=f^x_t(G').$
\end{proof}

\noindent
\textbf{Contraction of an arc in $S(t)$.}
Now, assume that $(u,v) \notin \Phi, (u,v) \in S(t)$. The contraction of $(u,v)$ consists in replacing every arc $(w,u) \in \delta^{\text{in}}_u$ by an arc $(w,v)$ of length $l'_{wv}=l_{wu}+l_{uv}$ and by removing the vertex $u$ and all its outgoing arcs. The weight of $u$ in $G$ is moved to the weight of $v$ in $G'.$ Namely, the weight of $v$ in the new graph is $w'_v = w_v+w_u \exp(-l_{uv}).$  Let $G'$ be the graph obtained from $G$ by contracting  $(u,v).$ The next lemma establishes that the contribution of $t$ to ECA in $G$ is equal to its contribution in $G'.$

\begin{figure}[!ht]
  \centering
  \subfloat[]{
    \includegraphics[width=0.45\linewidth]{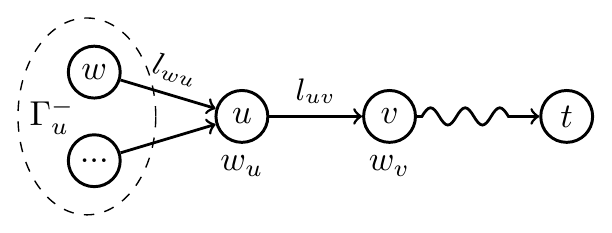}
  }
  \hfill
  \subfloat[]{
    \includegraphics[width=0.45\linewidth]{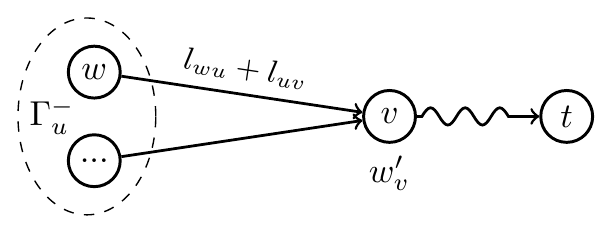}
  }
  \caption{(a) A graph $G$ before contraction of an arc $(u,v).$ (b) A graph $G'$ obtained from $G$ by contracting arc $(u,v)$. The weight $w'_v$ of $v$ in $G'$ is equal to $w_v+w_u \exp(- l_{uv}).$}
\end{figure}

\begin{lemma}\label{contract}
  Let $(u,v)\in S(t)$ and let  $G'$ be the graph obtained from $G$ by contracting arc $(u,v)$ and modifying accordingly the weight of $w_v.$ For every scenario $x\in \{0,1\}^\Phi,$  $f^x_t(G)=f^x_t(G').$
\end{lemma}

\begin{proof}
  Let $s$ be a vertex of $G$ and $x$ be a scenario in $\{0,1\}^\Phi$. We denote $f^x_{st}(G)$ the contribution to ECA of the pair $st$ with scenario $x\in \{0,1\}^\Phi$ in $G.$ If $s\notin \{u,v\},$ it is easy to check that, for any $x\in \{0,1\}^\Phi$ the length of the shortest path from $s$ to $t$ in $G$ is equal to the length of the shortest path from $s$ to $t$ in $G'.$
  %Indeed, if there exists a path $P$ of length $l$ in $G$ then there exists a path $P'$ of length $l$ in $G'.$ If $P$ passes through $u$ then $P'$ is obtained from $P$ by replacing the path consisting of the arcs $wu$ and $(u,v)$ by the arc $wv.$ Otherwise $P'=P.$  Conversely, if there exists a path $P'$ of length $l$ in $G$ then there exists a path $P$ of length $l$ in $G.$ If $P'$ contains an arc $wv$ which does not belong to $G$ then $P$ is obtained from $P'$ by replacing the arc $wv$ by the path of the same length consisting of the arcs $wu$ and $(u,v).$ Otherwise, $P=P'.$ 
  Moreover, the weight of $s$ is the same in $G$ and in $G'.$ Therefore, by the definition of ECA, $f^x_{st}(G)=f^x_{st}(G'),$ for any $x\in\{0,1\}^{\Phi}.$  On the other hand, the contributions of the pairs $ut$ and $vt$ in $G$ sum to  the contribution of $v$ in $G'.$ Indeed,
  \begin{eqnarray*}
    f^x_{ut}(G)+f^x_{vt}(G) & = & w_u \exp(- d_x(u,t)) + w_v \exp(- d_x(v,t))\\
    & = & w_u \exp(- (l_{uv} + d_x(v,t))) + w_v \exp(- d_x(v,t))\\
    %     & = & w_u\cdot \exp(- d_x(u,t)) + w_v \cdot \exp(-d_x(u,t)) \cdot \exp(l(uv))\\
    & = & (w_v+w_u \exp(- l_{uv})) \exp(- d_x(v,t))\\
    & = & w'(v) \exp(- d_x(v,t))\\
    & = & f^x_{vt}(G').
  \end{eqnarray*}
  We conclude that, for any $x\in \{0,1\}^\Phi,$  $f^x_t(G) = \sum_{s\in V} f^x_{st}(G) = f^x_{ut}(G) + f^x_{vt}(G) + \sum_{s\in V-\{u,v\}} f^x_{st}(G) = f^x_{vt}(G') + \sum_{s\in V-\{u,v\}} f^x_{st}(G') = f^x_t(G').$
\end{proof}

\noindent
\textbf{Graph reduction.}
Let $G_t$ be the graph obtained from $G$ by deleting every arc of $W(t)$ and contracting every arc of $S(t)$. By induction the contribution of $t$ is preserved. The preprocessing step consists in computing the graph $G_t$ for every vertex $t.$ For that, we compute $F(u,v)$ for every arc $(u,v)$ in $O(|A|\cdot(|A|+|V|\log |V|)).$ Then we transpose the representation to obtain $S(t)$ for each vertex $t.$ Similarly, we compute $W(t)$ for each vertex $t$ within the same complexity. The experiments of Section \ref{experiments} show that, when the number of arcs that can be protected is much smaller than the total number of arcs, replacing $G$ by $G_t$ in the construction of the mixed integer program reduces significantly the size of the MIP formulation and the running times.

  {
    \section{Greedy algorithms and their limits}\label{section_greedy_bad_cases}
  }

In the lack of an efficient method to compute an optimal solution of BC-ECA-Opt, ecologists often use greedy algorithms to compute sub-optimal solutions {\cite{OPTIMALITY-SCP_HANSON_2019}.
    In this section, we present four commonly used \textit{greedy algorithms} and highlight their pathological cases. In section \ref{optimality}, we compare the quality of the solutions obtained with these greedy algorithms with the optimal solution on four case studies.}

The \textit{Incremental Greedy} (IG) algorithm starts from the graph with no improved element. At each step $i$, the algorithm selects the element $e$ with the greatest ratio $\delta_e^i/c_e$ until no more element fits in the budget. Here, $\delta_e^i$ denotes the difference between the value of ECA with and without the improvement of the element $e$ at the step $i.$ As usual, $c_e$ is the cost of improving the element $e$. The element $e$ can be either an arc or a vertex.

The \textit{Decremental Greedy} (DG) {algorithm, similar to the Zonation algorithm \cite{ZONATION-ALGO_MOILANEN_2005},} algorithm starts from the graph with all improvements performed and iteratively removes the improvement of the element $e$ with the smallest ratio $\delta_e^i/c_e$. DG finishes with incremental steps to ensure there is no free budget left. These algorithms perform at most $|\Phi|$ steps and at each step $i$ need to compute $\delta_e^i$ for each element $e$. It is easy to implement IG in $O(|V|^3 + |\Phi|^2 \cdot |V|^2)$ by using an all pair shortest path algorithm to compute in $O(|V|^3)$ the initial distance matrix and then by performing $|\Phi|$ steps in which the computation of $\delta^i_e$ for each arc $e\in \Phi$ takes $O(|V|^2).$ Indeed, when we decrease the length of an edge we can update the distance matrix in $O(|V|^2)$.
For the implementation of DG, when we increase the length of an edge we cannot update the distance matrix as easily as for IG. We can recompute in $|V|^3$ the whole distance matrix for computing each $\delta^i_e$ and the algorithm runs in $O(|\Phi|^2 \cdot |V|^3)${. Dynamically updating shortest path lengths would improve computational complexity \cite{DYNAMIC-SHORTEST-PATHS_DEMETRESCU_2004}.}  These complexities are already too large for the practical instances handled by ecologists which can have few thousands of patches. Most of the studies using the PC or ECA indicators use simpler algorithms that we call \textit{Static Increasing} (SI) and \textit{Static Decreasing} (SD). These algorithms are variants of the greedy algorithms which do not recompute the ratio $\delta_e/c_e$ of each element $e$ at each step and thus are faster but do not account for cumulative effects nor redundancies.

  {Below, we provide} instances on which IG and DG performs poorly compared to an optimal solution. On these instances, it is easy to check that the solutions returned by SI and SD are not better than the solutions returned by IG and DG.
In the {following instances}, all arcs have a probability $1$ if improved and $0$ otherwise and have unitary costs. Recall that a spider is a tree
% (sorry for ecologists) 
consisting of several paths glued together on a central vertex {(Fig. \ref{bad_case_inc_graph}, \ref{bad_case_dec_graph} and \ref{bad_case_both_graph})}.

\begin{caseof}
  \case{Bad case for IG}{bad_case_inc}{The graph is a spider with $2k$ branches: $k$ long branches with two edges, an intermediate node of weight $0$ and a leaf node of weight $1$, and $k$ short branches consisting of a single edge with a leaf of very small weight $\epsilon > 0$, see Fig.2 (a). All branches are connected to a central node of weight $1$. IG performs poorly on this instance. Indeed, IG is tricked into purchasing short branches with very small ECA improvement because purchasing an edge of a long branch alone does not increase ECA at all. An optimal solution results in larger value of ECA by improving pairs of arcs of long branches.

    In this case, IG does not perform well while DG finds an optimal solution computed by the MIP solver except when the budget is $1.$ In this case, the reverse occurs: DG performs badly while IG is optimal.  Indeed, DG realizes that the budget is not sufficient to improve two arcs of a long branch only after removing the improvements of all short branches.

    \vspace{-0.5cm}
    \begin{figure}[H]
      \centering
      \subfloat[]{
        \raisebox{-0.525\height}{\includegraphics[width=0.3\linewidth]{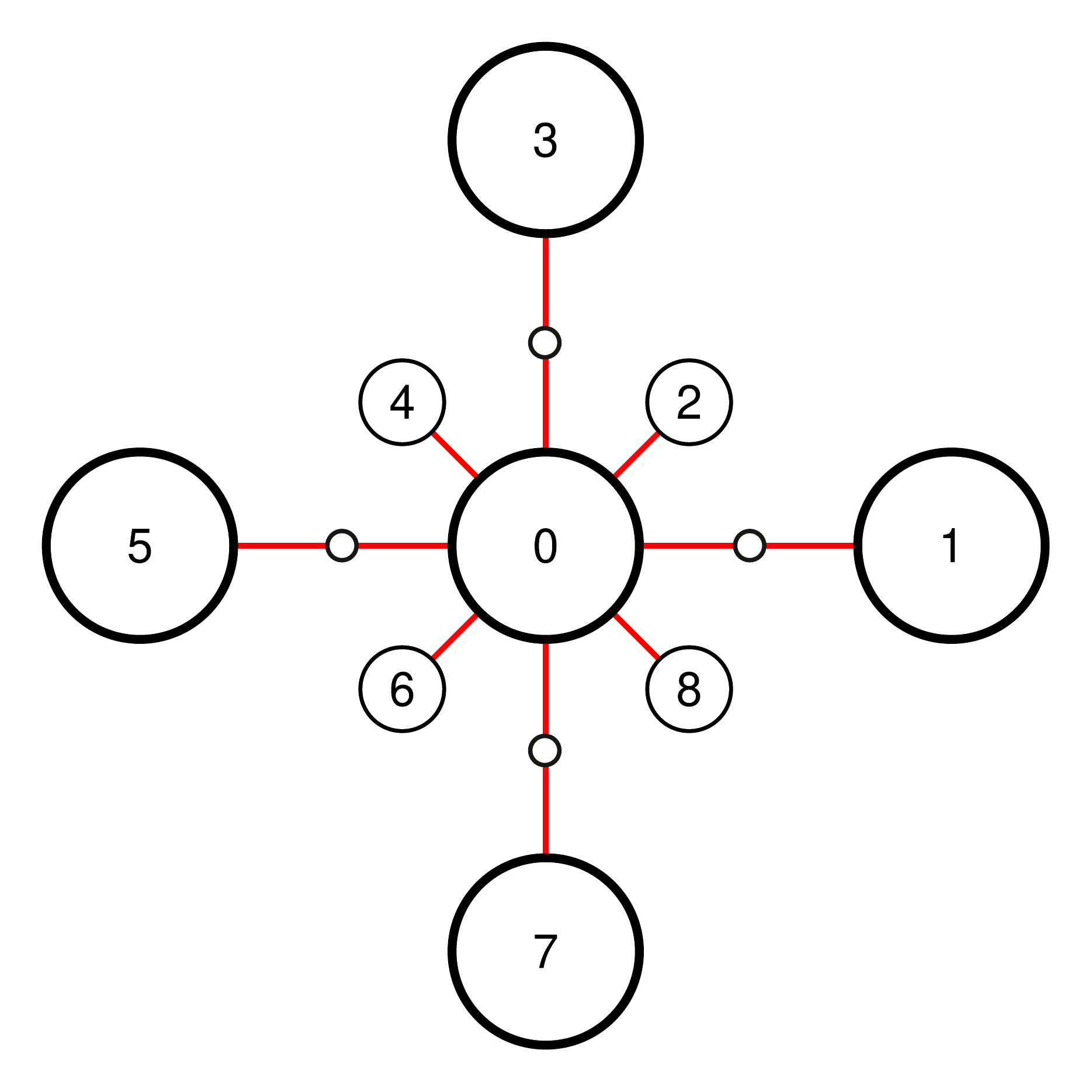}\label{bad_case_inc_graph}}}
      \hfill
      \subfloat[]{
        \raisebox{-0.5\height}{\includegraphics[width=0.675\linewidth]{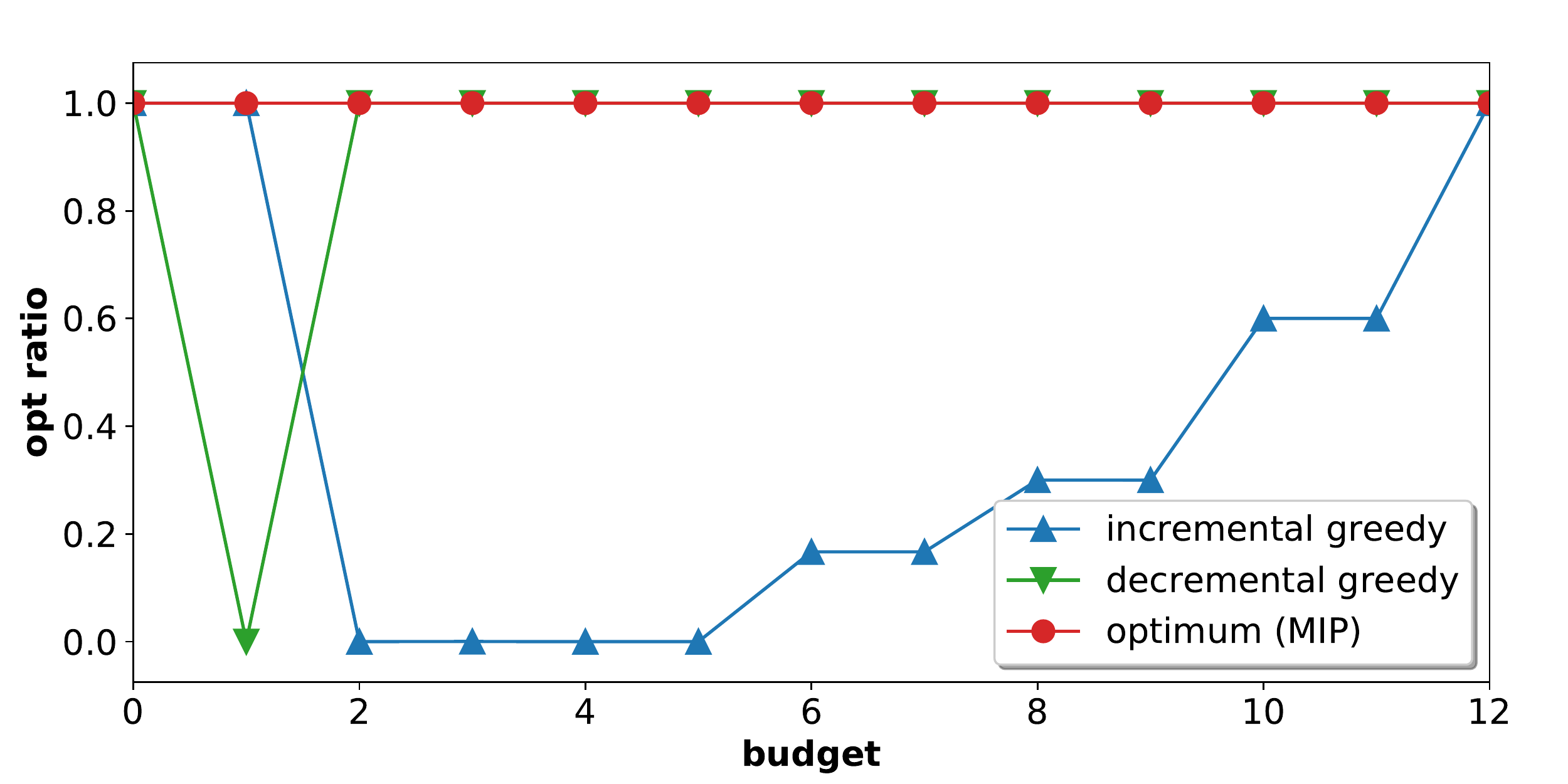}\label{bad_case_inc_plot}}}
      \caption{An instance on which Incremental Greedy fails. (a) the graph of the \hyperref[bad_case_inc]{{IG bad case}} with $k=4$, (b) ratio of the increase in ECA between the solutions returned by IG and DG and an optimal solution for several budgets.
      } \label{bad_case_inc_fig}
    \end{figure}
  }
  \case{Bad case for DG}{bad_case_dec}{The graph is obtained from a star with $k+1$ branches by replacing one branch by a path of  length $k.$ The central node and all leaves except the leaf of the path have weight $1.$ The leaf of the path has weight $1 + \epsilon$.  The internal nodes of the path have weight $0$. DG performs poorly on this instance because it removes one by one the branches of the star for which $\delta_e/c_e=1$ before removing an edge of the path for which $\delta_e/c_e=1+\epsilon.$ When the budget is at least $2,$ an optimal solution removes all the edges of the path before removing an edge of another branch.

    \vspace{-0.25cm}
    \begin{figure}[H]
      \centering
      \subfloat[]{
        \raisebox{-0.525\height}{\includegraphics[width=0.3\linewidth]{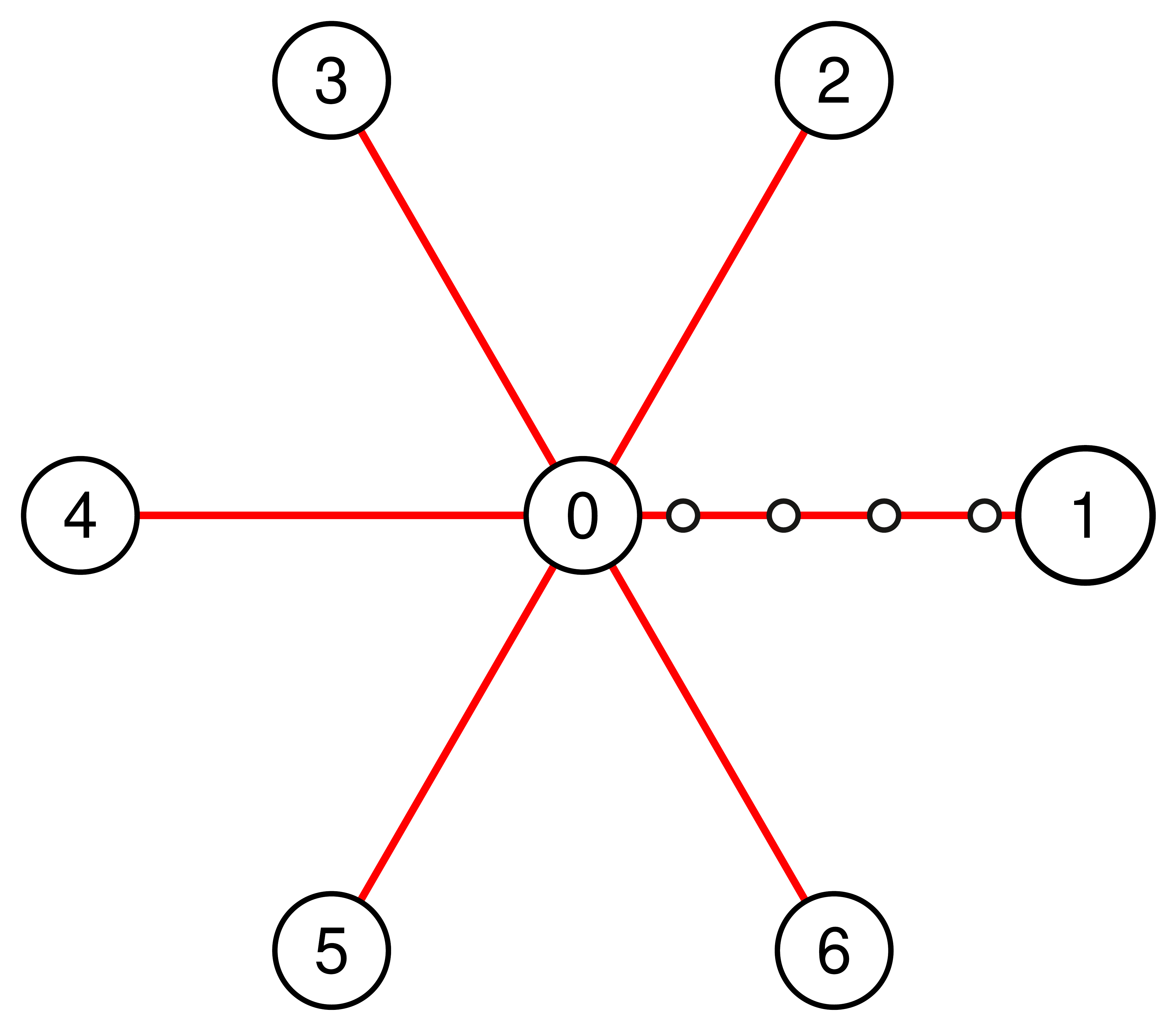}\label{bad_case_dec_graph}}}
      \hfill
      \subfloat[]{
        \raisebox{-0.5\height}{\includegraphics[width=0.675\linewidth]{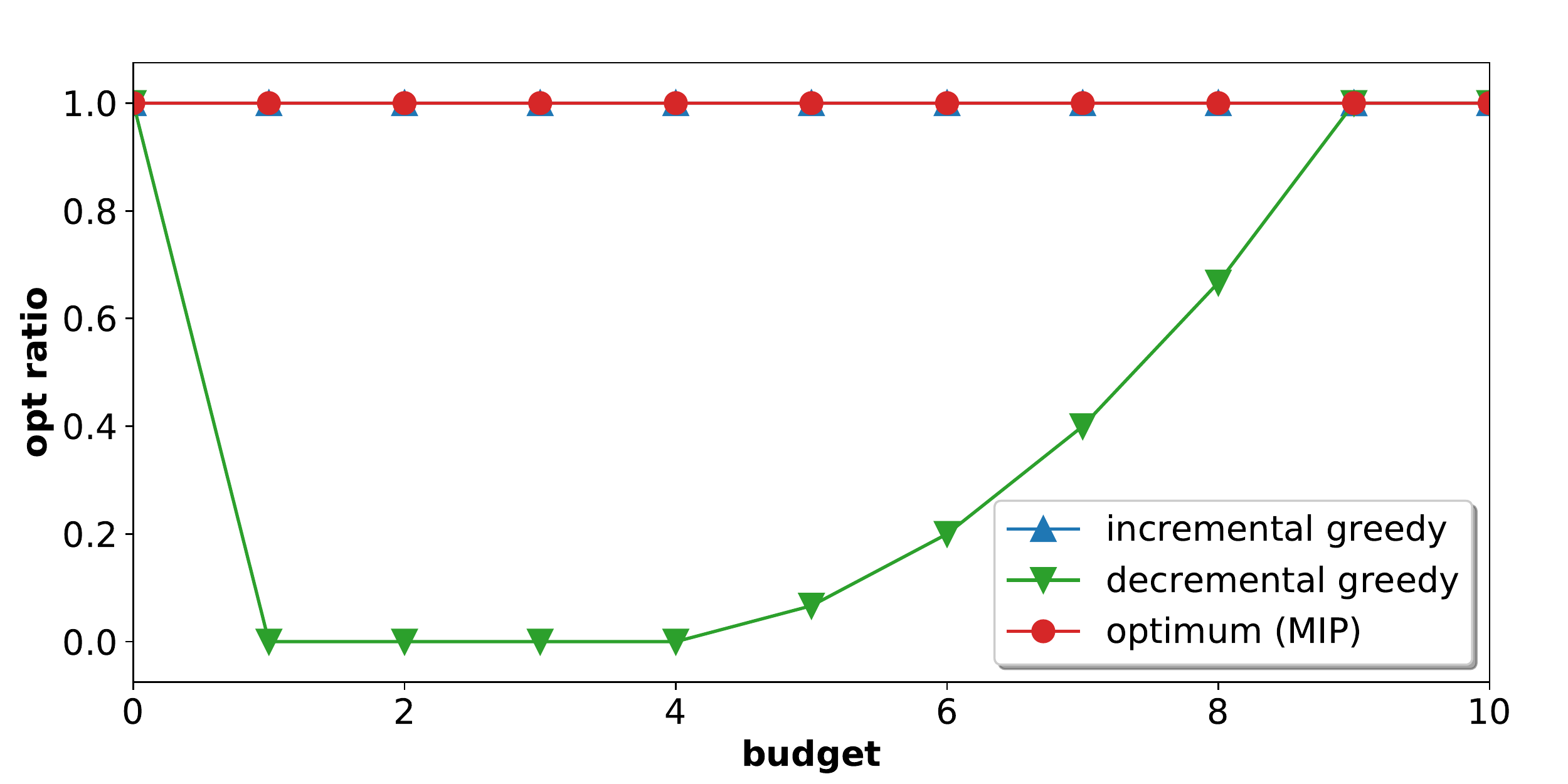}\label{bad_case_dec_plot}}}
      \caption{An instance on which Decremental Greedy fails. (a) the graph of the \hyperref[bad_case_dec]{{DG bad case}} with $k=5$, (b) ratio of the increase in ECA between the solutions returned by IG and DG and an optimal solution for several budgets.
      }\label{bad_case_dec_fig}
    \end{figure}
  }
  \case{Bad case for IG and DG}{bad_case_both}{The graph is a spider with $k+1$ branches. All branches except one are paths of length 2 with an internal node of weight $0$ and a leaf of weight $1.$ The last branch is a path of length $2k$ with internal nodes of weight $\epsilon > 0$ and a leaf of weight $1 + \epsilon$. All branches intersect in a central node of weight $1$. In this case, both Incremental and Decremental Greedy fail. On one hand, IG selects the edges of the path of length $2k$ one by one and does not realize that by taking two edges of a short branch it could improve much more ECA. On the other hand, DG removes first the edges of the short branch because the weight of  leaf of a long branche is $1+\epsilon$ while the weight of the leaf of a short branch is $1.$ Hence, DG and IG return the same law quality solution.

    \vspace{-0.25cm}
    \begin{figure}[H]
      \centering
      \subfloat[]{
        \includegraphics[width=0.8\linewidth]{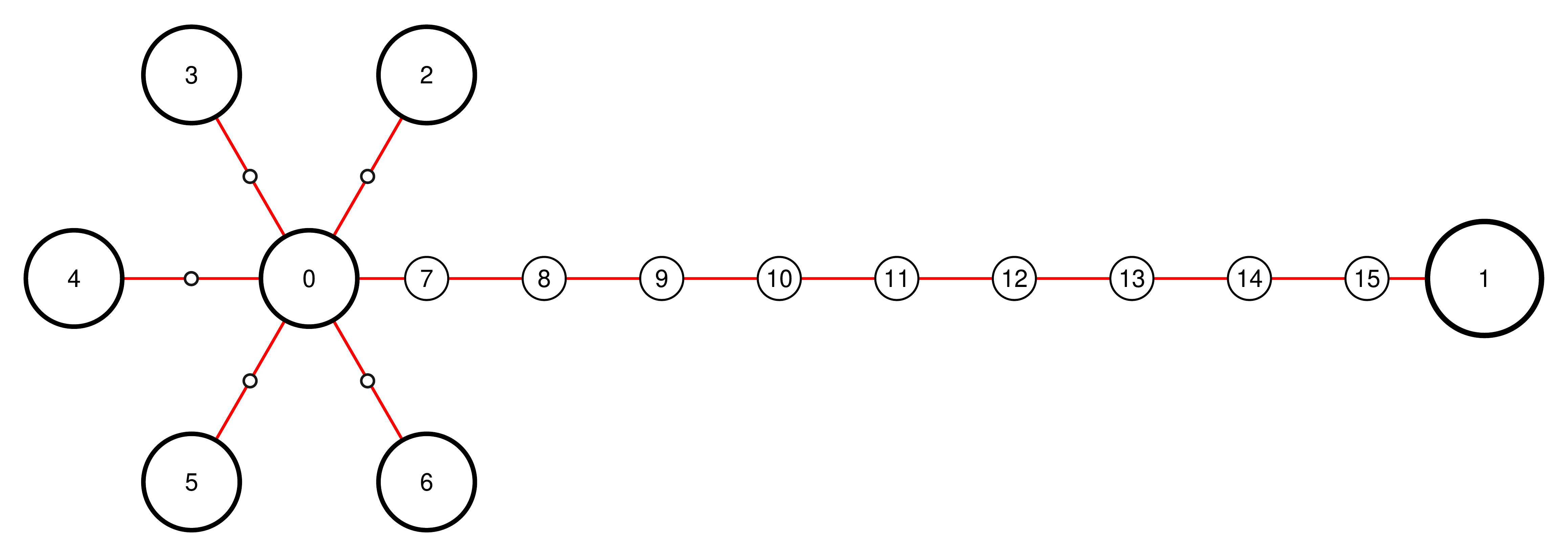}\label{bad_case_both_graph}}
      \hfill
      \subfloat[]{
        \includegraphics[width=0.9\linewidth]{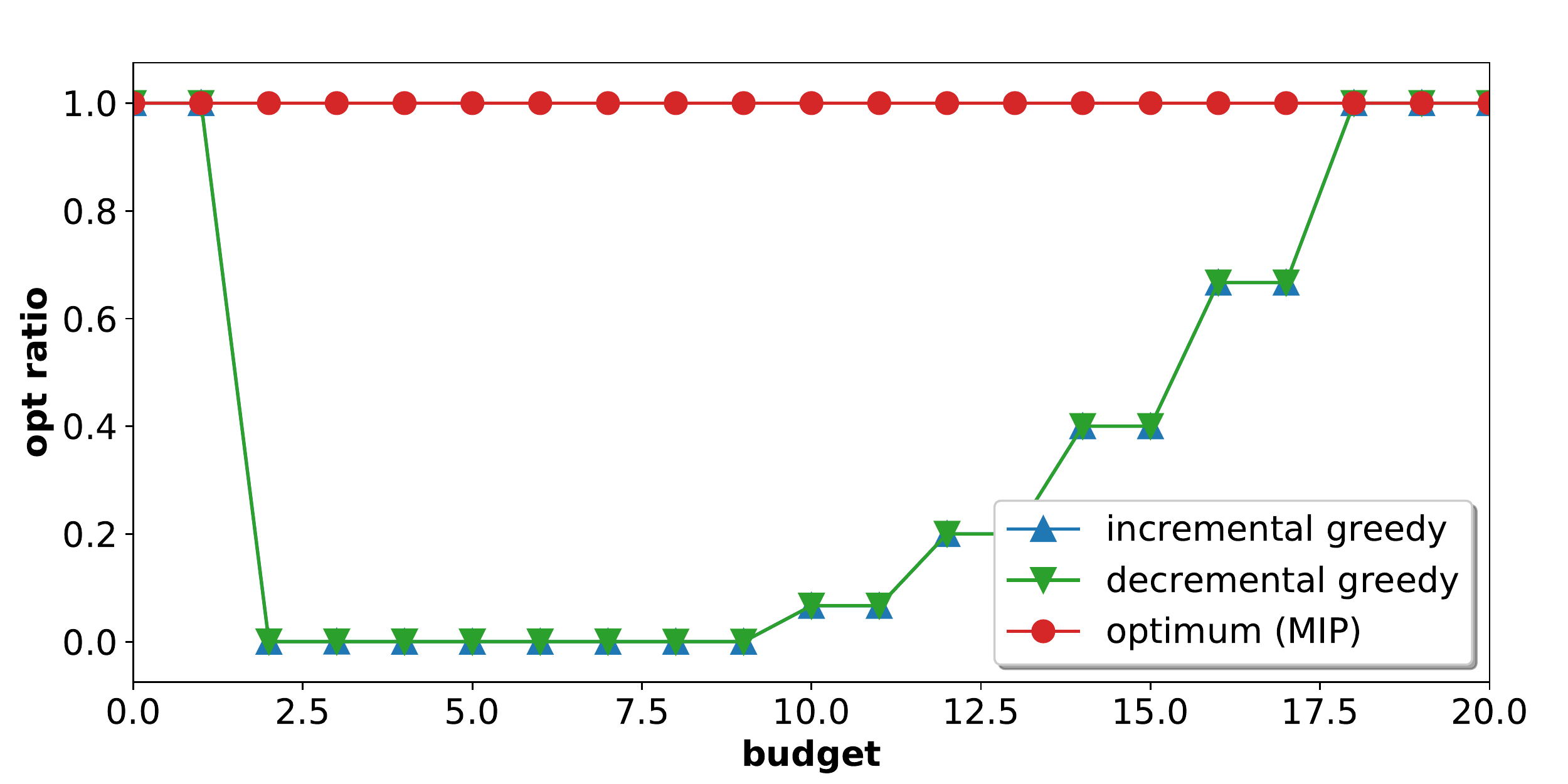}\label{bad_case_both_plot}}
      \caption{An instance on which both Incremental and Decremental Greedy algorithms fail. (a) the graph of the \hyperref[bad_case_both]{IG and DG bad case} with $k=5$, (b) ratio of the increase in ECA between the solutions returned by IG and DG and an optimal solution for several budgets.
        % Bad case for both incremental and decremental greedy algorithms. Nodes $0$, $2$, $3$, $4$, $5$ and $6$ are of weight $1$, node $1$ is of weight $1 + \epsilon$, nodes $7$ to $15$ are $\epsilon$ and others have zero weight. This bad case combines the bad behaviors of Fig \ref{bad_case_inc} and \ref{bad_case_dec} where the incremental greedy persists in constructing the connection between $0$ and $1$ not seeing the benifits of connecting $0$ to $2$, $4$, $6$ and $8$, while the decremental one destructs the connection from $0$ to $2$, $4$, $6$ and $8$ conserving the one between $0$ and $1$.
      }\label{bad_case_both_fig}
    \end{figure}
  }
\end{caseof}

The case of Figure \ref{bad_case_both_fig} illustrates the fact that IG and DG do not provide any constant approximation guarantee (even for trees), i.e. for any constant $0<\alpha<1$ there exists an instance of BC-ECA-Opt such that $ALG< \alpha OPT$ where ALG is the value of ECA for the best solution among those returned by IG and DG and $OPT$ is the value of ECA for an optimal solution.

\section{Numerical experiments}\label{experiments}

In this section, we report on our computational experiments in order to demonstrate the added benefit of our MIP formulation and preprocessing step. {We performed the numerical experiments} on a desktop computer equipped with an Intel(R) Core(TM) i7-8700k 4.8 gigahertz and 32 gigabytes of memory and running Manjaro Linux release 21.2.4. {We implemented our model as well as the preprocessing and greedy algorithms in C++17 using Gurobi Optimizer \cite{GUROBI} version 9.1.1 with default settings for solving MIP formulations, the graph library LEMON \cite{LEMON_GRAPH_2011} version 1.3.1 for managing graph algorithms, and the library TBB \cite{TBB_INTEL_2008} version 2020.3 for multithreading the preprocessing and greedy algorithms. Code is available at \url{https://gitlab.lis-lab.fr/francois.hamonic/landscape_opt_networks_submission}}.

% \textbf{Node spliting.} replacing a node $u$ in the graph with two nodes $u_{in}$ and $u_{out}$. Each outgoing arc $(u,v)$ of the original graph will correspond to an arc $(u_{out}, v)$ and each incoming arc $(v, u)$ will correspond to an arc $(v, u_{in})$. An additional bind arc $(u_{in}, u_{out})$ is then added.

\subsection{Instances} \label{instances}

{Below, we briefly describe the case studies on which we conduct experiments.
  % A vertex is said to represent a habitat patch if and only if its quality $w_u$ is strictly positive. 

  \textbf{Case study 1\label{river_case}} consists in identifying among a set of 15 dams present on the Aude river (France) those that need to be equipped with fish passes in order to restore the river connectivity for trouts \parencite{AUDE_SAINT-PE_2019}. The graph is a tree of $45$ vertices with $88$ arcs of which $30$ represent damns and can be improved by increasing their probability from $0$ to $0.8$.

  \textbf{Case study 2\label{quebec_case}} consists in identifying the remnant forest patches that need to be preserved from deforestation in the Montreal neighborhood (Canada) to guaranty habitat connectivity for the wood frog \parencite{FRAMEWORK_ALBERT_2017}. The graph is a planar graph of $59
    8$ vertices and $989$ arcs whose $260$ vertices can be improved by increasing their quality and $80$ arcs can be improved by increasing their probability from $0$ to $1$.

  \textbf{Case study 3\label{aix_case}} consists in identifying street sections in which planting trees can improve the connectivity of the urban canopy for the European red squirrel in the city of Aix-en-Provence. The graph is a triangular grid of $6186$ vertices and $27818$ arcs where $47$ street sections, each with an average of $90$ arcs, can be improved by increasing the probability of each arc $a$ from $\pi_a$ to $\pi_a^{1/6}$.

  \textbf{Case study 4\label{marseille_case}} consists in identifying wastelands that need to be preserved from artificialization to maintain connectivity among urban parks and the surrounding natural massifs in the city of Marseille for songbirds (e.g. Eurasian blackcap). The baseline graph is a near complete graph of $297$ nodes and $25024$ arcs, of which $100$ represent wastelands and can be improved by increasing their probability from $0$ to $1$.

  % Each vertex $u$ representing a wasteland is replaced by two vertices $u_{in}$ and $u_{out}$ which are respectively connected to the arcs entering $u$ and arcs leaving $u$, and an arc $(u_{in}, u_{out})$ of probability $0$. Preserving a wasteland is then modeled by increasing the probability of its corresponding arc from $0$ to $1$.
}

\subsection{Scalability and benefits of the preprocessing} \label{scalability}

{In this section we address the scalability of our approach and the added benefits of our preprocessing step. For this purpose we execute our method on about one hundred instances obtained from the four case studies by varying the budget.}

\begin{figure}[H]
  \centering
  \includegraphics[width=1\linewidth]{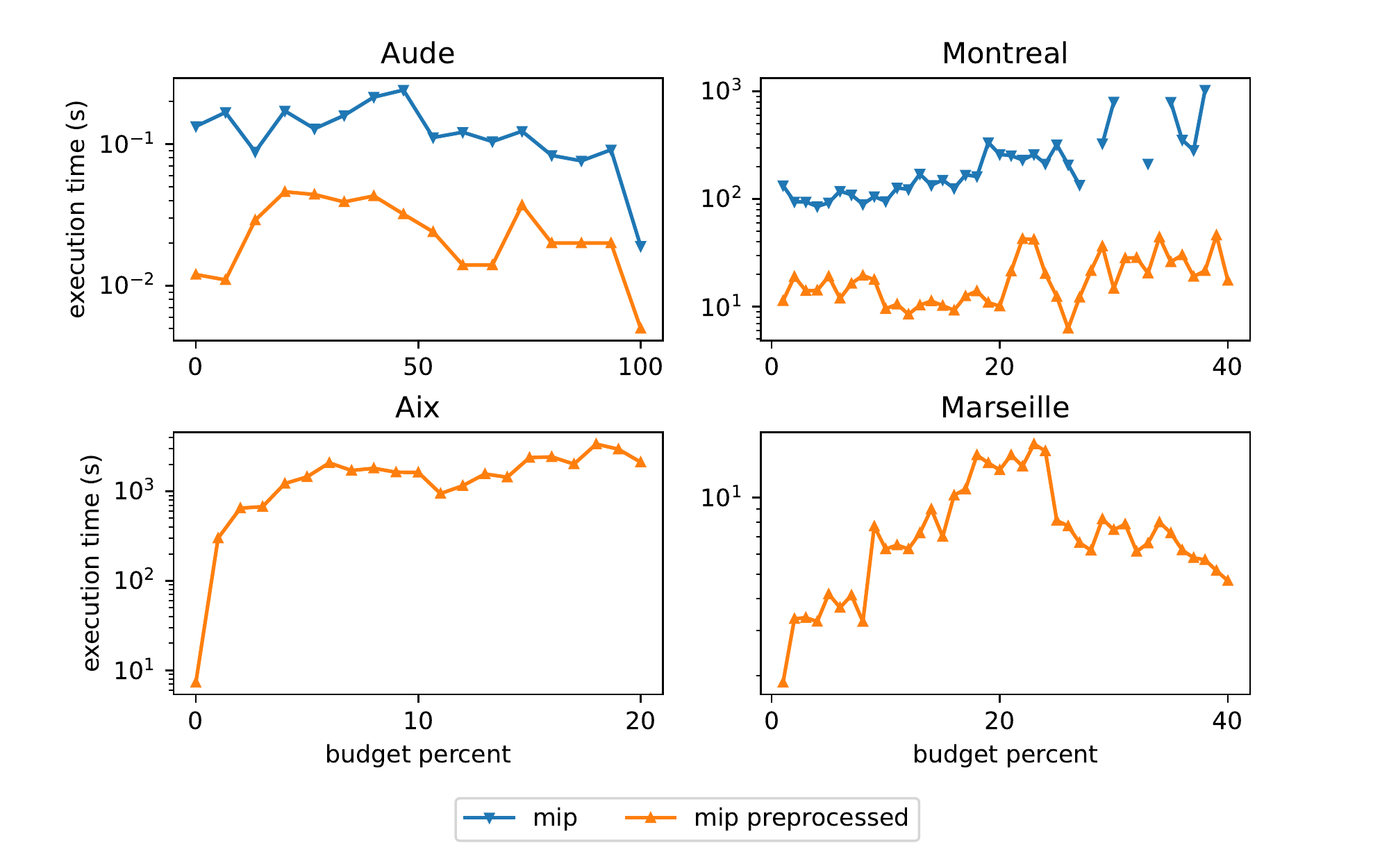}
  \vspace*{-0.6cm}
  \caption{Execution times on the four case studies as a function of the budget (missing points correspond to instances that do not finish within 10 hours)}
  \label{execution_times}
\end{figure}

{For the Aude and Montreal cases, the preprocessing reduces the resolution time by about 10 times (Fig. \ref{execution_times}). Without preprocessing, the Aix and Marseille instances are not solved by the optimizer within 10 hours, whereas with preprocessing they become solvable in about half an hour and half a minute respectively. In most unfinished instances, the optimizer reaches the optimal solution but is not able to complete the exploration of the search space in the allotted time.}

\begin{table*}[!ht]
  \centering
  \begin{tabular}{crrrrrrr}
    \multirow{2}{*}{case} & \multicolumn{3}{c}{MIP} & \multicolumn{4}{c}{preprocessed MIP}\tabularnewline
    \cmidrule(lr){2-4}\cmidrule(lr){5-8}
                          & \#var                   & \#const                                             & time   & \#var  & \#const & p. time & time \tabularnewline
    \hline
    Aude                  & 4061                    & 2551                                                & 120 ms   & 1069   & 1055    & 3 ms     & 20 ms\tabularnewline
    Montreal              & 830530                  & 262445                                              & 4 mins & 318848 & 167153  & 0.26 s    & 19 s\tabularnewline
    Aix                   & 1748708                 & 624841                                              & --     & 555124 & 295010  & $3$ s   & 1600 s\tabularnewline
    Marseille             & 4949825                 & 78410                                               & --     & 41676  & 22465   & $0.9$ s    & 7 s\tabularnewline
  \end{tabular}
  \caption{Comparison of the MIP and the preprocessed MIP according to the number of variables (\#var), the number of constraints (\#const), the preprocessing time (p. time) and the average {computation time} (time).}
  \label{pl_sizes_table}
\end{table*}

The preprocessing represents a small portion of the total {computation time} for all case studies (Table. \ref{pl_sizes_table}). The number of variables of the model is reduced by about $75\%$ in the Aude case, $60\%$ in the Montreal case, $70\%$ in the Aix case and $99\%$ in the Marseille case. This last number is explained by the fact that the Marseille graph is near complete and a large proportion of its arcs are $t$-useless for some vertex $t$. For constraints, the reduction is $60\%$ for the Aude case, $33\%$ for the Montreal one, $53\%$ for the Aix case and about $70\%$ for the Marseille case.

\begin{table*}[ht]
  \centering
  \begin{tabular}{crrrrrrr}
    \multirow{2}{*}{\#wasteland} & \multicolumn{3}{c}{MIP} & \multicolumn{3}{c}{preprocessed MIP} & \multicolumn{1}{c}{DG} \tabularnewline

    \cmidrule(lr){2-4}\cmidrule(lr){5-7}\cmidrule(lr){8-8}

                                 & \#var                   & \#const                              & time                                       & \#var  & \#const & time    & time \tabularnewline

    % \hline
    % 10 & 270623 & 13610 & 3094 & 5443 & 3700 & 161 & 0\tabularnewline
    \hline
    20                           & 345775                  & 18410                                & 12 s                                     & 10716  & 7197    & $< 1$ s     & $< 1$ s\tabularnewline
    % \hline
    % 30 & 446973 & 23810 & 36019 & 17611 & 11394 & 678 & 0\tabularnewline
    % \hline
    % 40 & 576713 & 29810 & 77326 & 25192 & 16225 & 1301 & 0\tabularnewline
    \hline
    50                           & 717901                  & 36410                                & 2 mins                                     & 34613  & 21485   & 2 s    & 7 s\tabularnewline
    % \hline
    % 60 & 888269 & 43610 & 313265 & 46176 & 27767 & 4843 & 0\tabularnewline
    % \hline
    % 70 & 1068105 & 51410 & 827349 & 60835 & 34768 & 7711 & 0\tabularnewline
    \hline
    80                           & 1306597                 & 59810                                & 32 mins                                    & 76281  & 42039   & 13 s   & 30 s\tabularnewline
    % \hline
    % 90 & - & - & - & 95037 & 50119 & 30227 & 0\tabularnewline
    % \hline
    % 100 & - & - & - & 116300 & 58179 & 45871 & 0\tabularnewline
    \hline
    110                          &                         & -                                    &                                            & 132225 & 66759   & 1 min   & 1 min 30 s\tabularnewline
    % \hline
    % 120 & - & - & - & 157120 & 76130 & 79866 & 0\tabularnewline
    % \hline
    % 130 & - & - & - & 179941 & 85837 & 142420 & 0\tabularnewline
    \hline
    140                          &                         & -                                    &                                            & 207999 & 96600   & 3 mins  & 3 mins\tabularnewline
    % \hline
    % 150 & - & - & - & 240414 & 108827 & 424784 & 0\tabularnewline
    % \hline
    % 160 & - & - & - & 274106 & 120402 & 697429 & 0\tabularnewline
    \hline
    170                          &                         & -                                    &                                            & 308355 & 132701  & 28 mins & 7 mins\tabularnewline
  \end{tabular}
  \caption{Comparison of the MIP, the preprocessed MIP and DG according to the number of variables (\#var), the number of constraints (\#const) and the time (on average with 20 different budget values) it takes to solve the Marseille instance with different numbers of wastelands}
  \label{nb_vars_table}
\end{table*}

We see in Table \ref{nb_vars_table} that the computation time of the MIP without preprocessing increases very quickly. It takes more than 30 minutes on average for instances with 80+ wastelands whereas the preprocessed one can be solved in less than 30 minutes with up to 170 wastelands.
This is due to the preprocessing step that significantly reduces the time required to solve the linear relaxation by reducing the number of variables, constraints and non-zero entries of the mixed integer program.
The preprocessed MIP is faster than the greedy algorithm for instances with at most 140 wastelands (the preprocessing step was not used for greedy algorithms).

\medskip

Finally, we run our preprocessing on 400 randomly generated instances from red a model of the landscape around Montreal for hares of $8733$ vertices and $18422$ arcs \cite{FRAMEWORK_ALBERT_2017} to study the impact of the preprocessing on the MIP formulation size. For building these instances, we take 20 connected subgraphs of 500 nodes and for each graph we create 20 instances by randomly picking a percentage of arcs whose probability could be increased from $\pi$ to $\sqrt{\pi}$.

% \begin{figure}[H]
%   \centering
%   \includegraphics[width=0.85\linewidth]{figures/"preprocessing benefits".pdf}
%   \caption{Average percentage of constraints, variables and non-zero entries that the preprocessing removes with respect to the percentage of arcs that could be restored}
%   \label{preprocess_benefits}
% \end{figure}

\begin{figure}[H]
  \centering
  \includegraphics[width=1\linewidth]{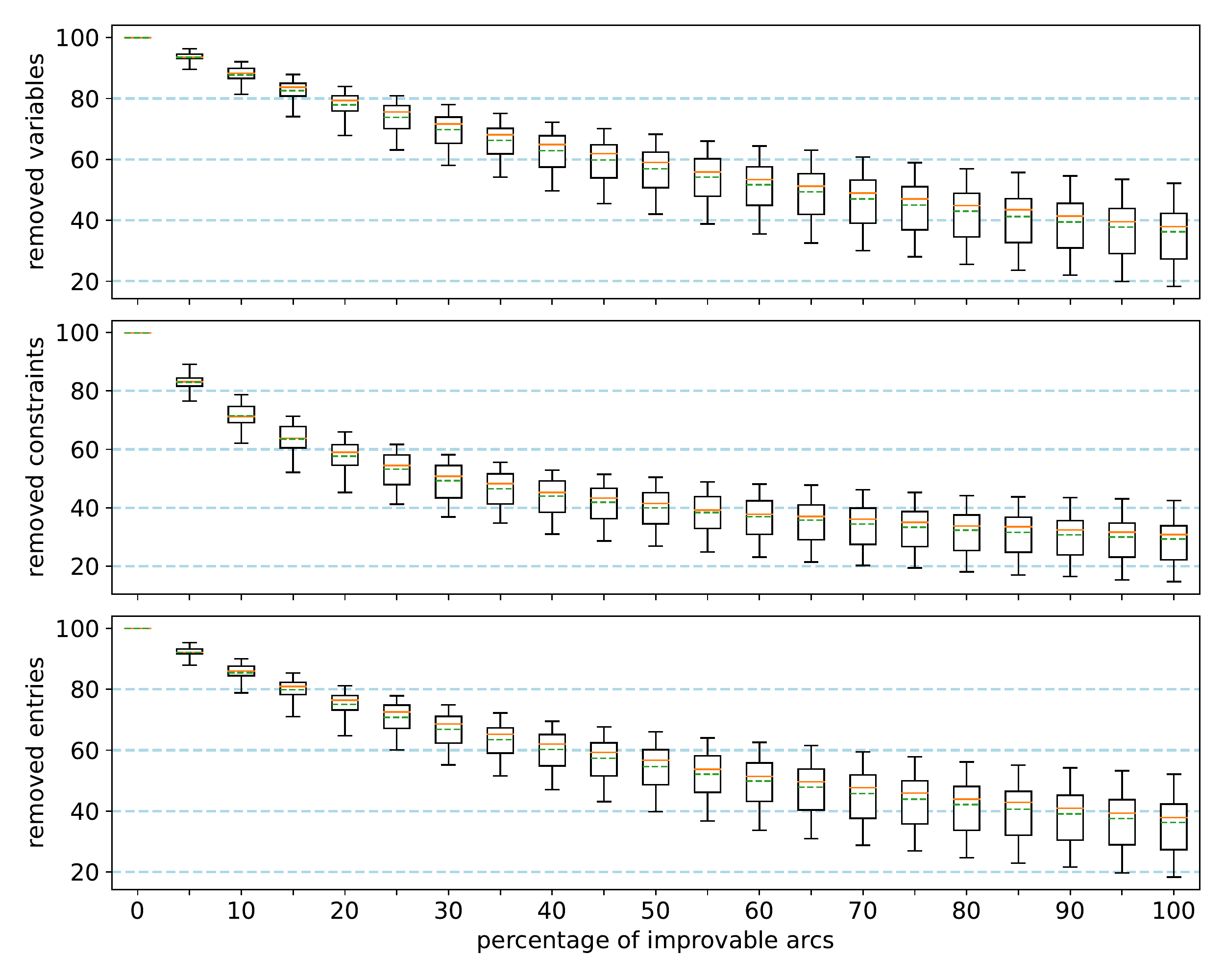}
  \caption{Box plots showing the percentage of constraints, variables and non-zero entries removed by the preprocessing as a function of the percentage of improvable arcs. The red line is the median, the dashed green line is the mean, the box represents the values between the 25th and 75th percentiles and the whiskers the min and max values.}
  \label{preprocess_benefits}
\end{figure}

Figure \ref{preprocess_benefits} shows {a box plot of}the percentage of reduction of the number of constraints, variables and non-zero entries of the MIP formulation with respect to the percentage of arcs that could be improved. The preprocessing removes almost all the elements of the MIP formulation when the number of improvable arcs arrives close to zero. This reduction decreases with the number of arcs that can be improved. When $20\%$ of the arcs can be improved, the preprocessing removes on average $80\%$ and at least $70\%$ of the model's variables, on average $60\%$ and at least $45\%$ of the model's constraints, and on average $75\%$ and at least $65\%$ of the model's non-zero entries. Even when $100\%$ of the arcs can be improved, the preprocessing reduces on average by $35\%$ the model's size. Since the gap between the 25th and 75th percentiles does not exceed $17\%,$ the reduction seems to be robust. These results look consistent with the one of Table \ref{pl_sizes_table}. Indeed, in the case of Aix, about $15\%$ of the arcs can be improved and our preprocessing reduces the number of variables by $70\%$ and the number of constraints by $65\%.$

%In Figure \ref{preprocess_benefits} we see that, in the type of instance we described, our preprocessing tends to remove all the elements of the MIP formulation when the number of restorable arcs tends to zero. At $20\%$ of restorable arcs, the preprocessing removed more than $75\%$ of the model's variables and non-zero entries and removed more than half of its constraints.

\subsection{Quality of the solutions} \label{optimality}

{
  \begin{figure}[H]
    \centering
    \includegraphics[width=1\linewidth]{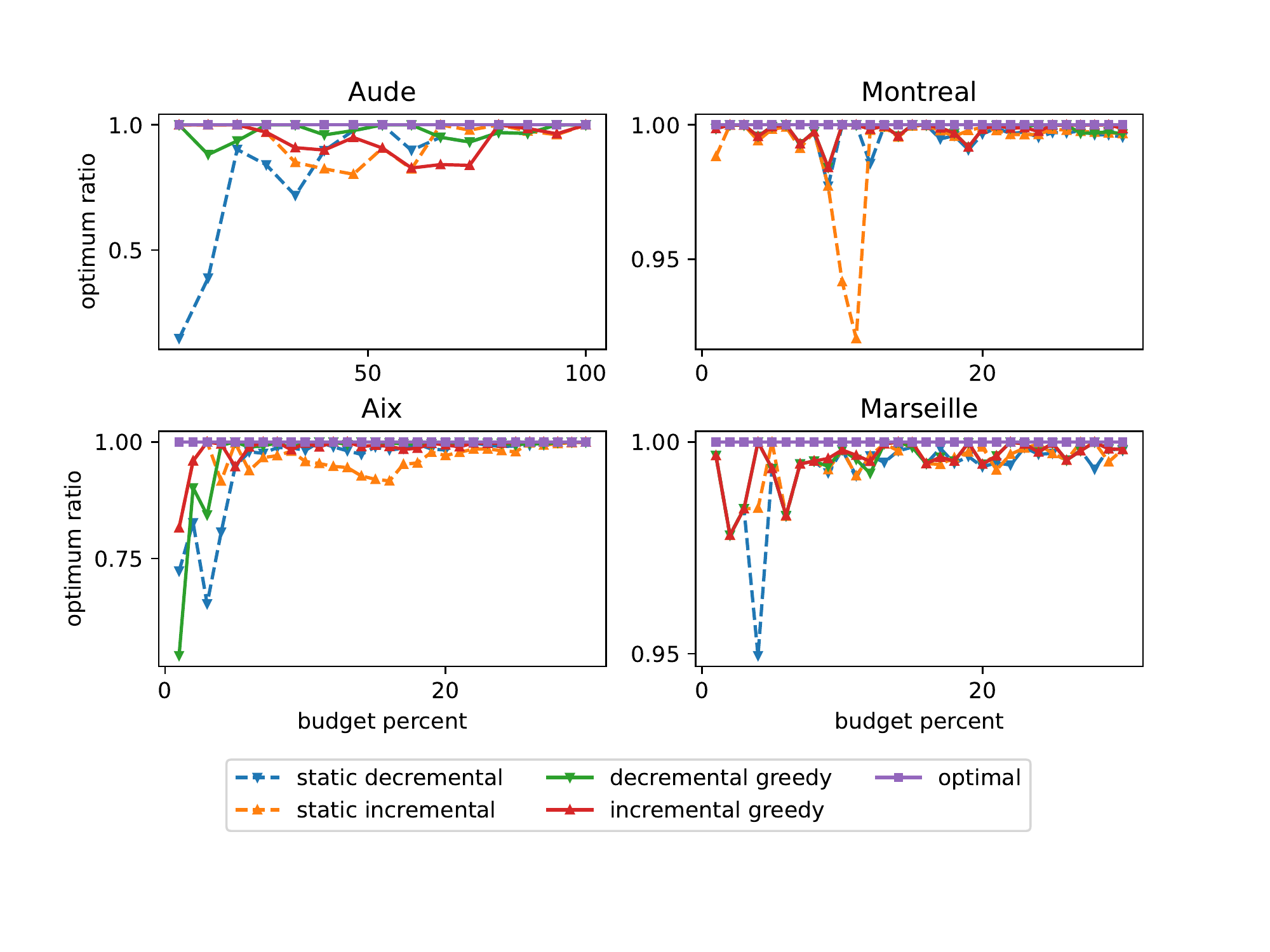}
    \vspace*{-1.7cm}
    \caption{Percentage gain in ECA achieved by the solutions of the different algorithms compared to the optimal solution for different budget values. \label{figure_optimum_ratios}}
  \end{figure}

  \begin{table*}[!h]
    \centering
    \begin{tabular}{crrrrrrrrr}
      \multirow{2}{*}{} & \multicolumn{2}{c}{IL} & \multicolumn{2}{c}{DL} & \multicolumn{2}{c}{IG} & \multicolumn{2}{c}{DG}\tabularnewline
  
      \cmidrule(lr){2-3}\cmidrule(lr){4-5}\cmidrule(lr){6-7}\cmidrule(lr){8-9}
  
                        & min.                   & avg.                   & min.                   & avg.                                  & min.    & avg.    & min.    & avg. \tabularnewline
      \hline
      Aude              & 80.3~\%                & 94~\%                  & 14.7~\%                & 83.9~\%                               & 82.8~\% & 94~\%   & 88.1~\% & 97.2~\%\tabularnewline
      \hline
      Montreal          & 92~\%                  & 99.2~\%                & 97.7~\%                & 99.6~\%                               & 98.4~\% & 99.8~\% & 98.4~\% & 99.8~\%\tabularnewline
      \hline
      Aix               & 81.6~\%                & 96.2~\%                & 65.2~\%                & 95.6~\%                               & 81.6~\% & 98.7~\% & 54.1~\% & 97.4~\%\tabularnewline
      \hline
      Marseille         & 97.8~\%                & 99.5~\%                & 95~\%                  & 99.3~\%                               & 97.8~\% & 99.6~\% & 97.8~\% & 99.6~\%\tabularnewline
      \hline
    \end{tabular}
    \caption{Minimum and average optimilaty ratio for each algorithm and case study.}
    \label{table_stats_algos}
  \end{table*}

  For each of the case studies and each of the four algorithms, there is at least one budget value for which the quality of the solution is significantly lower than the quality of the optimal solution, the greatest departures being observed at lower budget values (Fig. \ref{figure_optimum_ratios}, Table. \ref{table_stats_algos}). Greedy versions of incremental and decremental algorithms perform on average better than their static counterpart (Table. \ref{table_stats_algos}). The minimum and average optimality ratio in the Aude and Aix cases is lower than in the other cases, for all algorithms (Table. \ref{table_stats_algos}). Static and greedy algorithms are generally quite close to the optimal solution ($5\%$ lower on average). However, all algorithms, whether static or greedy, incremental or decremental, provide poor quality solutions for some budget values (Table. \ref{table_stats_algos}).}

\section{Conclusion}

This article introduces a new MIP formulation for BC-ECA-Opt and shows that this formulation allows to optimally solve instances having up to 150 habitat patches while previous formulations, such as those described in \cite{PLNE_XUE_2017} are limited to  30 patches. The preprocessing step reduces significantly the size of the graphs on which a generalized flow has to be computed, thus enabling to scale up to even larger instances.
%aims to contract the elements of the landscape whose contributions to ECA are proportionally related for any assignment of the arcs lengths within their bounds. 
We showed that this preprocessing step allows to greatly reduce the MIP formulation size and that its benefits increase when the proportion of arcs whose lengths can change decrease. This allows to tackle instances up to 300 habitat patches.

The optimum solutions obtained experimentallly are compared to the ones returned by several greedy algorithms. Interestingly, we found that greedy algorithms perform well in practice despite the arbitrary bad cases we spotted. {Therefore, greedy algorithms remains a reasonable choice for instances too large to be solved optimally by an MIP solver.}
Our next goals will be to experiment our approach on other practical instances of the problem arising from different ecological contexts.

On the theoretical side, we would like to investigate the problem from the point of view of an approximation algorithm. For instance, is it possible to find reasonable assumptions under which greedy algorithms are guaranteed to return a solution whose ECA value is at least a constant fraction of the optimal ECA? If these assumptions are fulfilled by the real instances that we considered, this would explain our experimental observations. Moreover, since a polynomial time approximation scheme has been given in the case of trees \parencite{APPROX-TREES_WU_2014}, it could also be interesting to know for which larger classes of graphs constant factor approximation algorithms for this problem exist. Another interesting question is to determine whether good solutions could be obtained by decomposing geographically the problem, by solving independently a subproblem for each region and then by reassembling the solutions. In this case, a notion of fairness could help to allocate the budget among the regions so that each region can enhance its own {\it internal connectivity} keeping a part of the budget to enhance the connectivity between the regions. Since ECA is based on the equivalence between a landscape and a patch, such a multilevel optimization approach looks promising.

%%%%%%%%%%%%%%%%%%%%%%%%%%%%%%%%%%%%
\section{Acknowledgments} % 100 words max
%%%%%%%%%%%%%%%%%%%%%%%%%%%%%%%%%%%%

We are grateful to the referees for a careful reading and many useful comments and suggestions. The research on this paper was supported by R\'egion Sud Provence-Alpes-C\^ote d'Azur, Natural Solutions and the ANR project DISTANCIA (ANR-17-CE40-0015). The Aix case study belongs to the Baum program (Biodiversity Urban Development Morphology) supported by the PUCA, the OFB and the DGALN.
For stimulating exchanges on the case studies, we also thank:
Patrick Bayle, %marseille
Simon Blanchet, %aude
Andrew Gonzalez, %quebec
Maria Dumitru, %quebec
J\'er\^ome Prunier, %aude
Bronwyn Rayfield, %quebec
Benoit Romeyer %aix
and Keoni Saint-P\'e. %aude

\printbibliography

\end{document}